\documentclass[english]{article}
\usepackage[T1]{fontenc}
\usepackage[utf8]{luainputenc}
\usepackage{geometry}
\usepackage{color}
\usepackage{babel}
\usepackage{float}
\usepackage{amsmath}
\usepackage{amsthm}
\usepackage{amssymb}
\usepackage{graphicx}
\usepackage{setspace}
\PassOptionsToPackage{normalem}{ulem}
\usepackage{ulem}
\usepackage[unicode=true,pdfusetitle,
 bookmarks=true,bookmarksnumbered=false,bookmarksopen=false,
 breaklinks=false,pdfborder={0 0 1},backref=false,colorlinks=false]
 {hyperref}

\makeatletter

\providecommand{\tabularnewline}{\\}

\theoremstyle{definition}
\newtheorem{defn}{\protect\definitionname}
\theoremstyle{plain}
\newtheorem{assumption}{\protect\assumptionname}
\theoremstyle{remark}
\newtheorem{rem}{\protect\remarkname}
\theoremstyle{definition}
\newtheorem{condition}{\protect\conditionname}
\theoremstyle{plain}
\newtheorem{prop}{\protect\propositionname}

\@ifundefined{showcaptionsetup}{}{%
 \PassOptionsToPackage{caption=false}{subfig}}
\usepackage{subfig}
\makeatother

\providecommand{\assumptionname}{Assumption}
\providecommand{\conditionname}{Condition}
\providecommand{\definitionname}{Definition}
\providecommand{\propositionname}{Proposition}
\providecommand{\remarkname}{Remark}

\begin{document}
\begin{doublespace}
\begin{center}
\textbf{\Large{}Options as Silver Bullets: Valuation of Term Loans,
Inventory Management, Emissions Trading and Insurance Risk Mitigation
using Option Theory}{\Large\par}
\par\end{center}

\begin{center}
\textbf{Ravi Kashyap (ravi.kashyap@stern.nyu.edu)}\footnote{Dr. Yong Wang, Dr. Isabel Yan, Dr. Vikas Kakkar, Dr. Fred Kwan, Dr.
William Case, Dr. Srikant Marakani, Dr. Qiang Zhang, Dr. Costel Andonie,
Dr. Jeff Hong, Dr. Guangwu Liu, Dr. Humphrey Tung and Dr. Xu Han at
the City University of Hong Kong; the editorial board, anonymous reviewers
and numerous seminar participants, particularly at a few meetings
of the econometric society and various finance organizations, provided
valuable suggestions to improve this paper. The views and opinions
expressed in this article, along with any mistakes, are mine alone
and do not necessarily reflect the official policy or position of
either of my affiliations or any other agency.}
\par\end{center}

\begin{center}
\textbf{Estonian Business School, Tallin, Estonia / Formation Fi,
Hong Kong / City University of Hong Kong, Hong Kong}
\par\end{center}

\begin{center}
\par\end{center}

\begin{center}
Change: Choices and Consequences
\par\end{center}

\begin{center}
AMS Subject Codes: 91G20 Derivative securities; 90B05 Inventory; 60G25
Prediction theory; 91B76 Environmental economics; 68U35 Computing
methodologies for information systems
\par\end{center}

\begin{center}
\begin{center}
\today
\par\end{center}
\par\end{center}

\begin{center}
\textbf{\textcolor{blue}{\href{https://doi.org/10.1007/s10479-022-04610-w}{Edited Version: Kashyap, R. (2022).  Options as Silver Bullets: Valuation of Term Loans, Inventory Management, Emissions Trading and Insurance Risk Mitigation using Option Theory.   Annals of Operations Research,  XX (S.I.: Business Analytics and Operations Research),  001-041.}}}\tableofcontents{}
\par\end{center}
\end{doublespace}
\begin{doublespace}

\section{\quad Abstract}
\end{doublespace}

\begin{doublespace}
Models to price long term loans in the securities lending business
are developed. These longer horizon deals can be viewed as contracts
with optionality embedded in them. This insight leads to the usage
of established methods from derivatives theory to price such contracts.
Numerical simulations are used to demonstrate the practical applicability
of these models. The techniques advanced here can lead to greater
synergies between the management of derivative and delta-one trading
desks, perhaps even being able to combine certain aspects of the day
to day operations of these seemingly disparate entities. These models
are part of one of the least explored, yet profit laden, areas of
modern investment management.

A heuristic is developed to mitigate any loss of information, which
might set in when parameters are estimated first and then the valuations
are performed, by directly calculating valuations using the historical
time series. This approach to valuations can lead to reduced models
errors, robust estimation systems, greater financial stability and
economic strength. An illustration is provided regarding how the methodologies
developed here could be useful for inventory management, emissions
trading and insurance risk mitigation. All these techniques could
have applications for dealing with other financial instruments, non-financial
commodities and many forms of uncertainty.
\end{doublespace}
\begin{doublespace}

\section{\label{sec:Our-Innovations-and}Introduction}
\end{doublespace}

\begin{doublespace}
The bulk of the existing studies on securities lending primarily focus
on the belief that activity in the securities lending markets can
be used to predict future security returns. Many existing studies
develop theoretical models and empirically test the corresponding
concepts on different public and proprietary data-sets. A quick survey
of existing studies on securities lending (Section \ref{sec:Fundamentals-and-Related})
makes it clear that there is hardly any paper that considers the motivations
of the main players, the actions that arise due to the incentives
the participants face and the impact of these actions on the securities
lending market. 

In this paper and related works, (Kashyap 2016a; 2016b) we attempt
to bridge this gap by deriving various results that consider the incentive
structure and the modus operandi of the players in the lending business.
Section (\ref{sec:Fundamentals-and-Related}) has a detailed review
of the literature. The references in this introductory section mainly
serve as a guide to familiarize the reader with this niche area of
finance and also seek to position this paper among other works in
this realm. Kashyap (2016a) has a recent and comprehensive coverage
of the literature on short selling and the stock loan space including
a background on securities lending. D’Avolio (2002); Jones \& Lamont
(2002); Duffie, Garleanu \& Pedersen (2002) have more details on the
historical evolution and the mechanics of the securities lending market.

Kashyap (2016a) has a detailed discussion of the organization of the
securities lending industry. The motivation of the main players and
the actions they undertake, including some new innovations, that could
lead to increased profitability are considered. Kashyap (2016a) also
looks at how to either design an appropriate securities lending exclusive
auction mechanism or to come up with a strategy for placing auction
bids, depending on which side of the fence a participant sits. These
two facets are dependent on whether the interest is: 1) to procure
the rights to use a portfolio for making stock loans such as for a
lending desk, or, 2) to obtain additional revenue from a portfolio
such as from the point of view of a long only asset management firm.
Kashyap (2016b) looks at a few other recent innovations being used
by lending desks towards establishing theoretical borrow rates, the
allocation of inventory to clients and estimating which securities
are likely to become harder to borrow. These two papers (Kashyap 2016a;
2016b) and the present paper are in a niche area of finance related
to securities lending. The contributions from these papers provide
a suite of methodologies that have wide application across the entire
spectrum of financial market participants. The corresponding insights
can be useful for both buy side and sell side institutions with potential
uses outside the financial landscape.

The following are some of the key contributions of this present paper
to the practice of investment management and to the wider set of tools
and methodologies in economics, finance and decision theory. Whenever
it is applicable, throughout the article, most of our results are
supplemented with practical considerations that can be operationally
useful on a daily basis. 
\end{doublespace}
\begin{enumerate}
\begin{doublespace}
\item We illustrate a novel application of derivative theory by considering
in detail the problem of long-term loans in the securities lending
business. We develop models to determine the loan rates on long term
loans. Section (\ref{subsec:Term-Loans-and}) has a discussion of
the fundamentals regarding how a term loan is structured and the need
for such a contract. Sections (\ref{subsec:Benchmark-Valuations};
\ref{subsec:Extremely-Exotic-Extensions}) have the formal theoretical
development.
\item We show that longer horizon deals can be viewed as contracts with
optionality embedded in them and priced using established methods
from option theory. This becomes, to our limited knowledge, the first
application that can lead to greater synergies between the operations
of derivative and delta-one trading desks\footnote{Delta one products are financial derivatives that have no optionality
and as such have a delta of (or very close to) one – meaning that
for a given instantaneous move in the price of the underlying asset
there is expected to be an identical move in the price of the derivative.
(\href{https://en.wikipedia.org/wiki/Delta_one}{Delta One,  Wikipedia Link})}. These two desks are usually organized as distinct business units
in most investment firms. The associated synergies could lead to the
direct use of risk management software, related tools and procedures,
and even personnel from derivative trading desks for delta-one desks.
This could perhaps even lead to combining certain aspects of the day
to day operations of these two seemingly disparate entities. Consolidation
of the information systems of these business units would bring down
costs significantly and also reduce the burden on personnel in terms
of having to use multiple information systems. Not to mention, additional
and diverse revenue sources are generated, which could be shared across
these desks. All of this will reduce business risks in various forms
including market risk, technology risk, information security risk,
among others. 
\item We run numerical simulations to demonstrate the practical applicability
of these models. Section (\ref{subsec:Sample-Data-Generation}) has
a discussion of how we generate the simulation data set and how we
endeavor to stay as close as possible to real data, which is difficult
to obtain. The complete real data set is only available to intermediaries
and most other market participants will have access to only some portion
of the dataset . Sections (\ref{subsec:Closed-Form-Benchmark}; \ref{subsec:Valuation-Matrix})
discuss the results of the simulations. Many well known properties
of option valuations are observed in the results, which confirm the
soundness of our techniques and the application of derivatives theory
for term loans. 
\item We develop a heuristic that can mitigate the loss of information that
sets in, when parameters are estimated first and then the valuation
is performed, by directly calculating the valuation using the historical
time series (Section \ref{sec:Lost-in-Estimation:}). To address and
mitigate the many recent concerns about economic and financial stability,
sophisticated models are being created and their parameters estimated
using historical data. By first doing an estimation and then using
the estimates for final calculations, we are potentially introducing
multiple levels of errors. Our technique is a way to perform calculations
directly using the historical time series. This ensures that not much
is lost in translation due to the need to first estimating the parameters
of a model or probability distribution and then using those to perform
model calculations. This can lead to reduced models errors and greater
financial / economic stability\footnote{\begin{doublespace}
\label{fn:We-summarize-current}We summarize current views on the
size and stability of the securities lending space below:
\end{doublespace}
\begin{enumerate}
\item The value of available inventory as of June 22, 2015, stands at \$13.22
trillion, according to a new info-graphic on the global securities
finance market from DataLend. Of the available inventory worldwide,
\$1.72 trillion was out on loan. The value of equity on loan was \$851
billion, while fixed income on loan stood at \$876 billion. Some 41,673
unique securities were out loan, according to the info-graphic, yielding
an estimated gross revenue of \$19.2 million per day on average, which
equates to \$2.26 billion for the first half of 2015. The US is still
the largest market with \$954 billion out on loan as of 22 June. Canada
is the closest market in size, with an estimated \$131 billion of
securities out on loan. Despite its size, the US commands a fee of
38 basis points (volume-weighted average, year to date), whereas Hong
Kong, which has \$28.8 billion out on loan, yields fees of 210 basis
points. (\href{http://www.securitieslendingtimes.com/securitieslendingnews/article.php?article_id=220006\#.VmahXL_O2iz}{Securites Lending Three})
\item As the potential risks of securities lending are discussed and debated
by the Financial Stability Oversight Council (FSOC), the U.S. Treasury’s
Office of Financial Research (OFR), and the Financial Stability Board
(FSB), it is important to try to understand both the overall size
of the securities lending market and the share of it attributable
to different participants. Based on one estimate from the FSOC the
percentage is typically around these values (Retirement and Pension,
Mutual Funds, Endowments, Insurance: 50\%, 35\%, 8\%, 6\%). (\href{https://www.ici.org/viewpoints/ci.view_14_sec_lending_02.print}{Securites Lending Four}).
\end{enumerate}
}.
\item We show how the techniques developed here could be potentially useful
for supply chain management (Section \ref{subsec:Inventory-Management}),
emissions trading (Section \ref{sec:Applications-to-Emissions}) and
insurance risk mitigation (Section \ref{subsec:Insurance-Risk-Mitigation}).
With globalization emerging as a permanent fixture of modern business,
our methodologies can be useful for negotiating contracts with international
partners, for making investment decisions and to promote socially
responsible behavior.
\item The immediate beneficiaries of the term loan techniques would be the
securities lending desks of sell side firms, since it would provide
them a theoretical basis for structuring term loans. Buy side firms
would benefit by entering into such contracts with sell side firms
and being able to lock down in advance what their cost of borrowing
would be for strategies that require short sales. Section (\ref{subsec:Buy-Side-Sell-Side-Perspective})
has a more detailed discussion on the benefits for the participants
involved.
\item These models are part of one of the least explored, yet profit laden,
areas of modern investment management. The next generation of models
and empirical work on securities lending activity would benefit by
factoring in the methodologies considered here. Sections (\ref{sec:Improvements-to-the};
\ref{sec:Conclusion}) suggest improvements and conclude.
\item To our limited knowledge, this is the first known instance of such
an application of options theory in the securities lending space,
for inventory management and the direct use of the historical time
series for model calculations.
\item In addition, the methodologies we have advanced have numerous applications
towards dealing with, and risk managing, several types of financial
instruments, non-financial commodities and many forms of uncertainty.
Hence, we could view the principal tools we have proposed, built using
option theory, as silver bullets for combating uncertainty and to
aid improved decision making.
\end{doublespace}
\end{enumerate}
\begin{doublespace}

\section{\label{sec:Fundamentals-and-Related}Fundamentals and Related Literature}
\end{doublespace}

\begin{doublespace}
While our study provides a direct application of option theory to
securities lending, there are many studies that consider the implicit
links between options prices and the short selling market. Evans et
al. (2009) examine short-selling constraints when options trade on
the underlying stock. Options market makers are effectively allowed
to sell short without borrowing the stock. By looking at the transactions
of a major options market maker, they find that in most hard-to-borrow
situations, the market maker, chooses not to borrow and instead fails
to deliver stock to its buyers. Battalio \& Schultz (2006) find no
evidence from apparent arbitrage opportunities that short-sale restrictions
prevented investors from shorting Internet stocks in the 1999 to 2000
period, using intra-day options data. They also show that investors
could have easily shorted stock synthetically by purchasing puts and
writing calls; investors can expect to receive almost as much from
a synthetic short sale as from an actual short. Battalio \& Schultz
(2011) examine how the September 2008 short sale restrictions impacted
equity option markets. They find that for options on banned stocks,
the trading costs (bid-ask spreads) increased dramatically. In addition
during the ban, synthetic share prices become significantly lower
than actual share prices, for banned stocks. They find similar results
for synthetic share prices of hard-to-borrow stocks, suggesting that
the dislocation in actual and synthetic share prices is attributable
to the increased hedging costs for options on banned stocks during
the short sale ban.
\end{doublespace}
\begin{doublespace}

\subsection{\label{subsec:Term-Loans-and}Term Loans and Optionality}
\end{doublespace}

\begin{doublespace}
Figure (\ref{fig:Securities-Lending-Term-Loan-Structure}) is a typical
term loan structure showing how an intermediary (securities lending
desk) sits between an inventory supplier (long only asset manager)
and final end borrowers (hedge funds, derivative traders, market makers,
etc.) who have short positions. The first portion of the figure (near
the circle marked one) shows the term loan contract arranged between
the intermediary and the end borrowers. Unlike regular stock loans,
term loans, as the name indicates, are decided for a fixed term when
the contract is initiated. Such long term loans can be arranged for
three to six months (other durations are also possible for some securities)
and the fees are fixed when the loan is made. The second half of the
figure (near the circle marked two) shows the intermediary sourcing
inventory from suppliers to make loans to the final end borrowers.
The final end borrowers pay a fixed fee for the term loan, unlike
regular loans for which the intermediary can change the fees on a
daily basis. Section (\ref{subsec:Buy-Side-Sell-Side-Perspective})
has details on the motivation for participating in such a long term
contract. Section (\ref{sec:Rainbows-and-Baskets}) considers the
valuation methodologies, including many alternate structures that
provide flexibility in terms of the amount of shares transacted and
the fees that apply, to arrive at a basis point estimate including
various assumptions that would be realistic from a securities lending
point of view. Definition (\ref{Defn:A-long-term}) follows from the
above discussion. We want to clarify that this definition is an adaptation
of the word loan, which is common in financial circles, to the specific
case of securities lending. 

\begin{figure}[H]
\includegraphics[width=17cm]{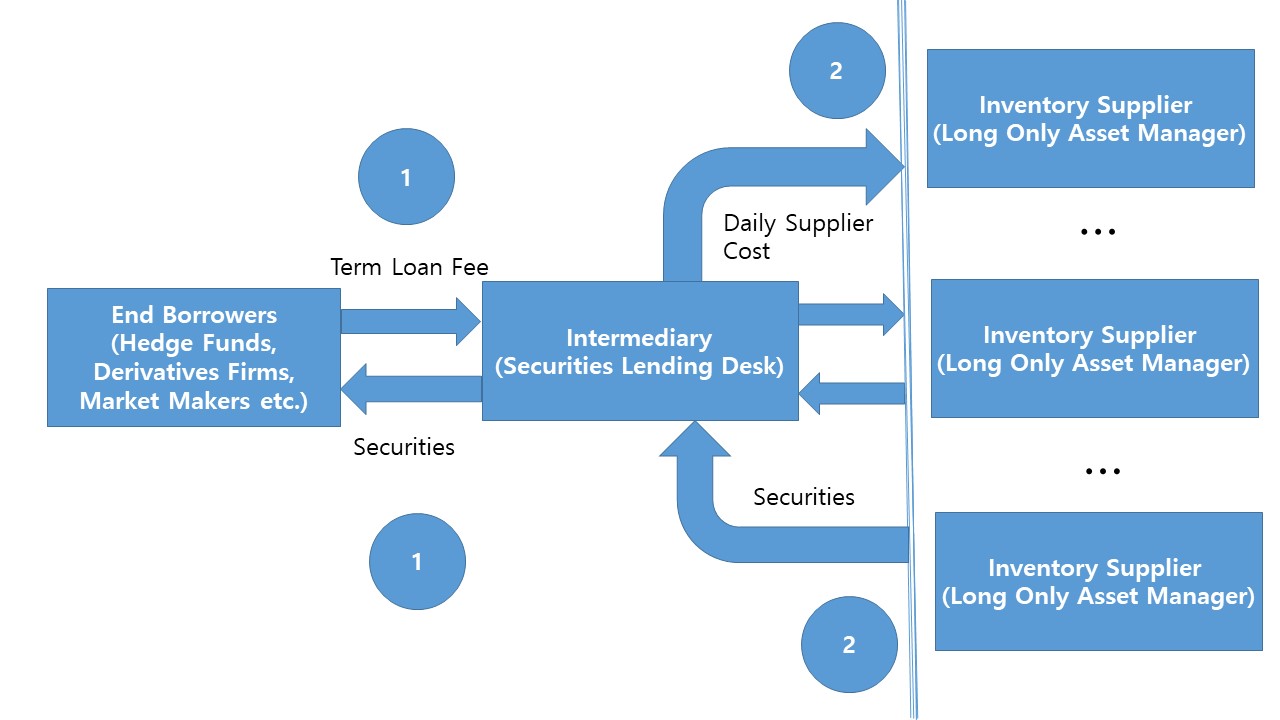}

\caption{\label{fig:Securities-Lending-Term-Loan-Structure}Securities Lending
Term Loan Structure}
\end{figure}

\end{doublespace}
\begin{defn}
\begin{doublespace}
\label{Defn:A-long-term}A long term loan is a structure wherein the
intermediary will guarantee a certain quantity to a short seller for
a certain time period at a fixed loan rate or within a band of rates. 
\end{doublespace}
\end{defn}
\begin{doublespace}
The guarantee is meant to imply that the intermediary will have to
provide the number of shares agreed upon in the term loan, but the
short seller can choose not to take on all the shares that are being
provided, highlighting the optionality present in the agreement. If
the intermediary is unable to meet the obligations of the term loan,
it will be liable for certain contractual commitments under the terms
of the term loan. The intermediary can guarantee a certain quantity
to a final end borrower by considering the fluctuations in the supply
of shares it receives, the rate at which it can source the shares
and the impact on profits when it is not able to meet the obligation
to fulfill the loan amount. This leads to the main assumption (\ref{Assumption: A-long-term})
of this paper which enables us to use pricing and risk management
tools from derivative theory to create term loan structures.
\end{doublespace}
\begin{assumption}
\begin{doublespace}
\label{Assumption: A-long-term}A long term security loan has optionality
due to the availability of shares being modeled as a Geometric-Brownian-Motion
(GBM). 
\end{doublespace}
\end{assumption}
\begin{doublespace}
Even though shares are traded, availability from a loan perspective
is a non-traded asset, making our model to price loans a ``pseudo
real option'' based methodology. Hence, we relate our methodology
to the use of options pricing for non-traded assets and briefly review
the literature on real options. The techniques in these studies can
be useful for extensions to our model, which adds to the growing use
of option theory in aiding decision making against various forms of
uncertainty. The links to real options will be clearer when we discuss
the inventory management application in Section (\ref{sec:Inventory-Management-Application}).

Bollen (1999); Adner \& Levinthal (2004); Trigeorgis (2005); Cuypers
\& Martin (2010) describe the distinction between real options and
the broader class of sequential decision-making processes including
examples of the use of real options and frameworks that explicitly
incorporate product life cycles. Lambrecht \& Perraudin (2003) discuss
preemption, cases when a firm fears that a competitor may seize an
advantage by acting first, under incomplete information. Trigeorgis
(1993) deals with the nature of option interactions and the valuation
of capital budgeting projects possessing flexibility in the form of
multiple real options. Tong \& Reuer (2007); Belderbos \& Zou (2009);
Loulianou et al. (2021) consider real options portfolio perspectives
on foreign affiliate divestments and risk implications within multinational
firms. Lee \& Makhija (2009) consider the effect of domestic uncertainty
on the real options value of international investments. Chi et al.
(2019); Trigeorgis \& Tsekrekos (2018) are detailed reviews about
real options in international business and operations research settings.

Quigg (1993) examines the empirical predictions of a real option pricing
model using a large sample of real estate market prices. Tee et al.
(2014) analyze the effects of an emissions trading scheme on the value
of bare-land on which pine trees are to be planted by applying a real
options method, assuming stochastic carbon and timber prices. Baldi
\& Trigeorgis (2015) develop a real options theory of strategic human
resource management based on human capital flexibility or adaptive
capability to respond to a range of future contingent landscapes in
contrast to the static traditional human resources view of employee
groups as a portfolio based on specificity and value. Grullon, Lyandres
\& Zhdanov (2012) close the loop between traded assets and real options,
by finding evidence that the positive relation between firm-level
stock returns and firm-level return volatility is due to firms’ real
options. Consistent with real option theory, they find that the positive
volatility-return relation is much stronger for firms with more real
options and that the sensitivity of firm value to changes in volatility
declines significantly after firms exercise their real options. Martínez-Ceseña
\& Mutale (2011); Fernandes, Cunha \& Ferreira (2011); Boomsma, Meade
\& Fleten (2012); Reuter et al. (2012); Liu \& Ronn (2020); Locatelli
et al. (2020) are studies regarding real options in the energy sector.
\end{doublespace}
\begin{doublespace}

\subsection{\label{subsec:Buy-Side-Sell-Side-Perspective}Buy Side and Sell Side
Perspective}
\end{doublespace}

\begin{doublespace}
The sell side here would be the collection of intermediary firms that
source supply and lend it on to final end borrowers. The buy side
here would have two segments of firms. One, the end borrowers who
either have a proprietary trading strategy or hedging that requires
shorting certain securities. Two, the beneficial owners who are long
securities and provide supply to the intermediaries also fall under
the buy side category. Depending on which side a firm falls under,
they will find the below derivations useful, since it will affect
the rates they charge or the rates they pay.

The primary beneficiaries of long term loans would be the actual short
sellers that have trading strategies dependent on being implemented
for a certain time horizon. If they are able to short securities at
reasonable loan rates without getting recalled or the rate getting
hiked, their trading strategies are more likely to be profitable.
Clearly, they would be willing to pay more for term loans, since the
rates on regular loans can be changed on a daily basis and the number
of shares on loan can either be reduced or the loan can be closed
out forcibly (known as a recall). The intermediaries that make term
loans will find this as a lucrative new revenue stream that eliminates
some of the corresponding volatilities. A similar reasoning on the
inventory supply side gives the following term loan structure.
\end{doublespace}
\begin{rem}
\begin{doublespace}
An alternate structure, or a long term borrow, that locks in the rate
and amount the intermediary borrows from external suppliers can also
be easily priced using our methodology.
\end{doublespace}
\end{rem}
\begin{doublespace}
Such a structure can ensure that the intermediary can lock in a minimum
level of profits on securities with volatile supply or if they are
expected to become hard to borrow. It is worth noting a stark difference
between an exclusive contract and a long term borrow. In an exclusive
contract, one intermediary will have access to all the long positions
of a beneficial owner or supplier; but the owners can still take back
their shares at any time or demand a higher borrow rate. Term loans
are less common on the supply side since inventory owners are not
willing to lock up their positions for extended periods of time. The
administrative hassle of having to put in place contracts for individual
securities, periodic assessments of being able to uphold those contracts
and other ongoing maintenance concerns for the beneficial owners makes
exclusive contracts more popular on the supply side.

Term loans have been offered by short selling desks for at-least the
last ten to eleven years. Though, to our awareness, there is no rigorous
work that provides a pricing methodology or connects the creation
and ongoing maintenance of term loans to established financial-economic
principles. Baklanova, Copeland \& McCaughrin (2015); Kashyap (2016a);
Footnote (\ref{fn:We-summarize-current}) have more details on the
size of the securities lending market. Though, there is also no data
available on the size and profitability of term loans. The global
securities on loan is around 2 trillion USD (Figure \ref{fig:Securities-Lending-Market}).
More than 10\% of the securities in all regions have loan rates in
excess of 5\% annually and there are securities with loan rates of
almost 25\%, indicating that there could be strong drivers for both
term loans and borrows, from both sides of the market. Perhaps, part
of the obstacle for the further development of the term loan business
could be the lack of a more technical approach and the training of
the personnel on the desk regarding option pricing and risk management.

\begin{figure}[H]
\includegraphics[width=10cm]{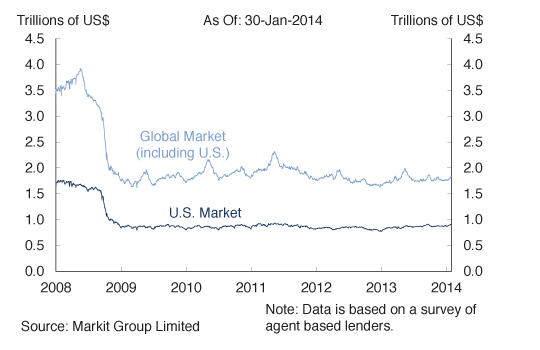}\textbf{\uline{USA}}\includegraphics[width=7cm]{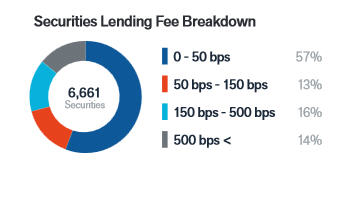}

\textbf{\uline{Asia}}\includegraphics[width=7cm]{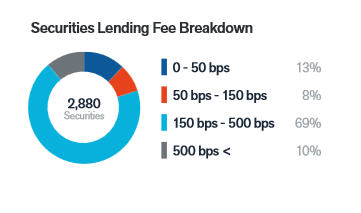}\textbf{\uline{Western
Europe}}\includegraphics[width=7cm]{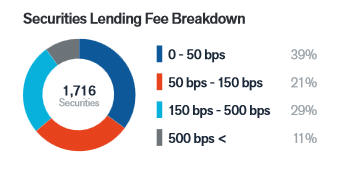}

\caption{Securities Lending Market Size and Loan Fees\label{fig:Securities-Lending-Market}}
\end{figure}

\end{doublespace}
\begin{doublespace}

\subsection{\label{subsec:Uncertainty-and-Unintended}Uncertainty and Unintended
Consequences}
\end{doublespace}

\begin{doublespace}
As we will see in Sections (\ref{sec:Rainbows-and-Baskets}; \ref{subsec:Setting-the-Stage};
\ref{subsec:Benchmark-Valuations}; \ref{subsec:Extremely-Exotic-Extensions})
valuation of a term loan requires understanding uncertainty from numerous
angles. As the participants try to find better and improved ways to
capture this uncertainty (Kashyap 2017), we might see that the profitability
of using this mechanism might decrease for participants from both
sides. This can lead to us believe that over time, as better valuation
methods are used by the participants, in an iterative fashion, the
profits will continue to erode. The other effect might be that with
more participants doing term loans, a secondary market for the options
on a term loan could get established. Also perhaps, with increasing
number of term loans, the availability of shares will get locked in
causing exponential pressure on the rates. A welcome outcome might
be the mitigation or hedging of the risk or exposure from the rest
of the loan book by reducing the total variance of profits. Froot
(1995) examines the hedging properties of real assets\footnote{\begin{doublespace}
Real assets increase in price in response to inflation shocks, with
some similarities to our earlier discussion on real options.
\end{doublespace}
}, with respect to a portfolio of stocks and bonds and finds that leveraged
positions in commodities with a high energy component, such as oil,
exhibit strong hedging properties by reducing the total variance significantly.

The cyclical nature of the transactions, which in some case can have
its tentacles spread far and wide, can result in catastrophic repercussions,
especially when huge sums of money move back and forth (Kashyap 2017).
Uboldi (2016) looks at the distortions (and perhaps instability) in
the prices of agricultural goods caused by the huge volumes of derivatives
traded on them. No discussion involving randomness is complete (Taleb
2005; 2007), especially one involving randomness to the extent that
we are tackling here, without being highly attuned to spurious results
mistakenly being treated as correct and extreme situations causing
devastating changes to the expected outcomes. Things can go drastically
wrong even in simple environments Sweeney \& Sweeney (1977), hence
in a complex valuation of the sort that we are dealing here, extreme
caution should be the rule rather than the exception. Kashyap (2017)
looks at recent empirical examples related to trading costs where
unintended consequences set in. With the above background in mind,
let us look at how we can model a term loan and come up with an indicative
price.
\end{doublespace}
\begin{doublespace}

\section{\label{sec:Rainbows-and-Baskets}Rainbows and Baskets of Binary Barrier
American Options}
\end{doublespace}

\begin{doublespace}
A simple term lending structure can be modeled as a Binary American
cash or nothing / asset or nothing put option. This option pays a
fixed cash rebate or the payoff of another asset if a certain binary
condition is satisfied during the life of the option. The binary condition
occurs when the underlying price falls below the strike price or a
preset barrier is breached making this a Binary Barrier down-and-in
put option. For a term loan, the underlying process is the availability
of shares being modelled as a GBM (Assumption \ref{Assumption: A-long-term}).
The event when the availability falls below a certain threshold, which
in our scenario is the number of shares that are being loaned out,
is the binary condition when the barrier is breached. (Section \ref{subsec:Benchmark-Valuations})
develops this analogy using all the variables that arise when pricing
such an option.

By using this insight, we can use established ways from derivative
pricing theory to price and even possibly hedge such deals. Windcliff
et al. (2007) discuss hedging with correlated assets. As a comparison,
it is worth noting that the basic structure will be cheaper than the
corresponding vanilla put option and we will report this benchmark
in the numerical results. Additional structures can be created by
using other combinations of European, American, Asian, call, put,
up, down, in, out, etc. and even on multiple securities, Rainbows
and Baskets, depending on the preferences and willingness of loan
desks and their clients. The possibilities are endless.

Barrier options are among the most common types of exotic options.
They are prevalent enough to almost consider them as plain simple
options (Carr 1995). The techniques introduced in the seminal paper
by Black \& Scholes (1973) can be applied for the valuation of barrier
options. Merton (1973) is the first work to present the valuation
of barrier options with the example of the down-and-out European call
option. Carr (1995) presents a brief but essentially complete survey
of the literature on barrier option pricing along with two extensions
of European up-and-out call option valuation. Rich (1991); Rubinstein
\& Reiner (1991) provide closed form solutions for a variety of standard
European barrier options (i.e. calls or puts which are either up-and-in,
up-and-out, down-and-in, or down-and-out). Using a conventional Black-Scholes
option-pricing environment, Hui (1996), obtains analytical solutions
of one-touch double barrier binary options that include features of
knock-out, knock-in, European and American style options. Nunes et
al. (2020) look at the early exercise boundary of American-style double
knock-out options.

Karatzas \& Wang (2000) obtain closed-form expressions for the prices
and optimal hedging strategies of American put-options in the presence
of an “up-and-out” barrier (also, Ingersoll 1998), both with and without
constraints on the short-selling of stock, demonstrating the close
links between option theory and need for securities lending. Zvan,
Vetzal \& Forsyth (2000) present an implicit method for solving partial-differential-equation
models of contingent claims prices with general algebraic constraints
on the solution. Examples of constraints include barriers and early
exercise features. In this framework, barrier options with or without
American-style features can be handled in the same way. Either continuously
or discretely monitored barriers can be accommodated, as can time-varying
barriers.

Haug (2001) uses the reflection principle (Harrison 1985) to provide
closed form valuation of American barrier options. In a barrier context
(e.g. a down-and-in call) the reflection principle basically states
that the number of paths leading from the stock price, $S_{t}$, at
a particular point in time, $t$, to a point higher than $X$, that
touch a barrier level $H\;(H<S_{t})$ before maturity, is equal to
the number of paths from an asset that starts from $H^{2}/S_{t}$
and that reach a point higher than $X$. Using the reflection principle
we can then simply value both European and American barrier options
on the basis of formulas from plain vanilla options. Thavaneswaran,
Appadoo \& Frank (2013) use fuzzy set theory (Zimmermann 1996; Carlsson
\& Fullér 2001) to price binary options since traditional option pricing
models determine the option’s expected return without taking into
account the uncertainty associated with the underlying asset price
at maturity.

Brockman \& Turtle (2003) propose a framework for corporate security
valuation based on path-dependent, barrier option models instead of
the commonly used path-independent approach by arguing that path dependency
is an intrinsic and fundamental characteristic of corporate securities
because equity can be knocked out whenever a legally binding barrier
is breached. A direct implication of this framework is that equity
will be priced as a down-and-out call option. This deviates from the
work of Black \& Scholes (1973), after which corporate securities
have been viewed in terms of standard (i.e., non-exotic) options written
on the underlying assets of the firm.

Margrabe (1978) started the theory of rainbow (multi-asset) options
by evaluating the option to exchange one asset for the other at expiry.
Stulz (1982) provides analytical formulas for European put and call
options on the minimum or the maximum of two risky assets. Ouwehand
\& West (2006) derive the Black–Scholes prices of several styles of
rainbow options using change of numeraire machinery. Hucki \& Kolokoltsov
(2007) develop a general approach for the pricing of rainbow options
with fixed transaction costs from a game theoretic point of view,
with applications for the framework of real options. Alexander \&
Venkatramanan (2012) derive general analytic approximations for pricing
European basket and rainbow options on multiple assets. Chen, Wang
\& Wang (2015) study the valuation and hedging problems of forward-start
rainbow options. We will point out other useful references in the
relevant sections below.
\end{doublespace}
\begin{doublespace}

\subsection{\label{subsec:Setting-the-Stage}Setting the Stage}
\end{doublespace}

\begin{doublespace}
We first discuss some properties of the different variables and how
the intermediary looks to influence them to increase profits. We define
all the variables as we introduce them in the text but (Appendix \ref{sec:Dictionary-of-Notation})
has a complete dictionary of all the terminology and symbols used
in the main results. The total amount of regular loans made, $L_{it}$,
other than the term loans on a security, has to meet the criteria
in Condition (\ref{cond:regular-loans}) with five numbered inequalities.
Before $L_{it}$ can satisfy the inequality on the subsequent level
it must be lesser than the earlier combination of the following five
variables: internal inventory the intermediary holds, $I_{it}$; amount
used from any exclusives, $E_{it}$; amount available from the exclusives,
$P_{it}$; external borrows, $B_{it}$; and $O_{it}$, any other supply
that can be sourced externally. 
\end{doublespace}
\begin{condition}
\begin{doublespace}
\label{cond:regular-loans}1) $L_{it}\leq I_{it}$ ; 2) $L_{it}\leq I_{it}+E_{it}$
; 3) $L_{it}\leq I_{it}+E_{it}+P_{it}$ ; 4) $L_{it}\leq I_{it}+E_{it}+P_{it}+B_{it}$
; 5) $L_{it}\leq I_{it}+E_{it}+P_{it}+B_{it}+O_{it}$.
\end{doublespace}
\end{condition}
\begin{doublespace}
Should there be a shortage, the amount on regular loans will be used
towards fulfilling obligations on terms loans. This should make it
clear that the total availability, $A_{it}=I_{it}+E_{it}+P_{it}+B_{it}+O_{it}$.
This excludes amounts on other terms loans on the same security. Should
there be a shortage, we can create a pecking order of term loans depending
on the profitability / risks on a specific loan and factors related
to the counter party, in terms, of the value of the business being
carried out with that counter party by the desk and even other parts
of the business. We assume that the process for total availability,
$A_{it}$, follows a GBM (Eq: \ref{eq:1}). It is possible to assume
that the individual components follow processes of their own and find
the properties of the combined process. We make this simplifying assumption
and take the availability as exogenously driven. Pirjol \& Zhu (2016)
is a good resource regarding the sum of GBMs. 

Stock prices and the borrow rates are also taken as exogenous GBMs
(Eq: \ref{eq:1}; \ref{eq:1-A}). What happens in practice is that
there is usually a baseline for the loan rates that is derived from
a combination of the borrow rate and a theoretical rate. The theoretical
rate is used to moderate the loan rate in case the borrow rate is
completely absent or is stale (no recent borrow for many days or even
weeks sometimes). It is also useful when supply becomes available
cheaply from some sources and remains expensive from other sources.
Weighted average borrow rates or the latest borrow rate are also used
sometimes. The theoretical rate and the spread added to the borrow
rate to form the loan rate are other decision variables at the disposal
of the intermediary. 

For simplicity, we take the spread to be a proportion of the borrow
rate, that is, $q_{it}=cQ_{it}$. $q_{it}$, is the spread added to
the borrow rate, $Q_{it}$, to form the loan rate $R_{it}$. $c,$
is the constant factor governing the spread and the borrow rate (Eq:
\ref{eq:2}). This assumption is the most realistic scenario, but
depending on the size of the exclusive and internal inventory (indicative
of market share and hence pricing power), the loan rates can further
be taken as variables the intermediary can influence. A deeper discussion
of how loan rates are set including the addition of a spread component
will be taken up in subsequent papers (Kashyap 2016b) devoted to just
the complex mechanics of rate manipulations, where we relax this assumption
and consider a wide array of factors that can alter the spread.

The external borrows, external supply, the internal inventory and
exclusive holdings represent number of shares, and hence are always
positive making them good candidates to be modelled as GBMs. The borrow
process is highly volatile, with the order of magnitude of the change
in the total amount of shares lent out, over a few months, being multiple
times of the total amount. The internal inventory and external supply
can change significantly as well, though there would be less turnover
compared to the borrow process. This would of course depend on which
parts of the firm the inventory is coming from. The holdings of the
exclusive are the least volatile of the three processes that govern
shares (or at-least the intermediary would hope so). The volatility
of inventory turnover (or any supply) can be a sign of the quality
of the inventory and this can be used to price a rate accordingly.
This extension and other improvements, where the loan rates and the
internal inventory can be made endogenous as opposed to the present
simplification, where they are exogenous, will be considered in a
subsequent paper (Kashyap 2016b). $S_{it}$ is the security price
at a particular time, $t,$ until the next time period, $t+1$. Here,
subscript $i$ denotes the $\text{i}^{th}$ security in a portfolio
and the number securities ranges from $i=1\;to\;i=n$. $n$ is the
number of securities available in the term deal, $i\in\left\{ 1,\,...\,,n\right\} $.
This applies when we are looking at rainbows, baskets or other multi-security
structures.

\begin{align}
\text{Geometric Brownian Motion } & \equiv\begin{cases}
\frac{dS_{it}}{S_{it}} & =\mu_{S_{i}}dt+\sigma_{S_{i}}dW_{t}^{S_{i}}\\
\frac{dA_{it}}{A_{it}} & =\mu_{A_{i}}dt+\sigma_{A_{i}}dW_{t}^{A_{i}}\\
\frac{dQ_{it}}{Q_{it}} & =\mu_{Q_{i}}dt+\sigma_{Q_{i}}dW_{t}^{Q_{i}}\\
\frac{dI_{it}}{I_{it}} & =\mu_{I_{i}}dt+\sigma_{I_{i}}dW_{t}^{I_{i}}\\
\frac{dO_{it}}{O_{it}} & =\mu_{O_{i}}dt+\sigma_{O_{i}}dW_{t}^{O_{i}}\\
\frac{dB_{it}}{B_{it}} & =\mu_{B_{i}}dt+\sigma_{B_{i}}dW_{t}^{B_{i}}\\
\frac{d\left(E_{it}+P_{it}\right)}{\left(E_{it}+P_{it}\right)} & =\mu_{\left(E_{i}+P_{i}\right)}dt+\sigma_{\left(E_{i}+P_{i}\right)}dW_{t}^{\left(E_{i}+P_{i}\right)}
\end{cases}\label{eq:1}
\end{align}
\begin{eqnarray}
\text{Geometric Brownian Motion} & \Longleftrightarrow & \text{Log\;\ Normal\;\ Processes}\label{eq:1-A}\\
W_{t}^{X_{i}} & \Longleftrightarrow & \text{Weiner\;\ Process\;\ governing}\;X_{i}^{th}\;\text{variable}.\nonumber \\
E(dW_{t}^{X_{i}}dW_{t}^{X_{j}}) & = & \rho_{X_{i},X_{j}}dt\;\quad=0\nonumber \\
\rho_{X_{i},X_{j}} & \Longleftrightarrow & \text{Correlation\;\ between}\;W_{t}^{X_{i}}\:and\:W_{t}^{X_{j}}\nonumber \\
X_{i} & \in & \left\{ S_{i},Q_{i},A_{i},I_{i},O_{i},B_{i},\left(E_{i}+P_{i}\right)\right\} \nonumber 
\end{eqnarray}
\begin{equation}
R_{it}=Q_{it}+q_{it}=Q_{it}+cQ_{it}=\tilde{c}Q_{it}\label{eq:2}
\end{equation}
The locate process\footnote{Another piece of the puzzle is the locate requests received by the
lending desk on a daily basis. These locate requests are sent by end
borrowers, in advance of actually borrowing shares to short, to get
an indication of the quantity of shares they can borrow. This is done
to ensure that their shorting needs for the trading day can be met.
The intermediary can fill either a portion or the entire locate request
depending on its inventory situation and also depending on how many
firms are sending it locates for that particular security for that
trading day. But once a locate request is filled by a lending desk,
they are expected to have that number of shares ready for the borrowing
firm. A borrowing firm, on the other hand, can borrow as much of the
filled locate amount as it chooses to. This mismatch between locate
approvals and actual borrows then leads to another aspect of the lending
business that can be optimized by implementing different variations
of the Knapsack Algorithm (Martello \& Toth 1987) and we consider
this in another paper (Kashyap 2016b). The conversion factor from
locates to borrows can be estimated as part of the locate approval
optimization. For the present purpose of estimating a term loan value,
we take this conversion factor as exogenously given. Lending desks
have been considering charging a nominal fee based on the locate amount
they agree to fill to discourage borrowers from sending in spurious
locate requests, though this practice is yet to be formally institutionalized
across the lending industry.} is more precisely modelled as a Poisson process since it would be
reasonably accurate to consider locates as discrete events occurring
in time (Eq: \ref{eq:3}; \ref{eq:4}) . That is, requests for a certain
number of shares being received in a given time interval. Given that
most of the time, the number and size of the share requests can be
large, we would need to use a high value of the arrival rate, $\lambda_{i}$.
Hence, we approximate this poison process as the absolute value of
a normal distribution with appropriate units (Cheng 1949; Eq: \ref{eq:5}).
This introduces a certain amount of skew, which is naturally inherent
in this process. The locate process can be useful to know what demand
the desk can expect and when higher demand is anticipated, it can
supplement the volatility of the availability, as an indicator that
the rate on the term loan needs to be higher. $D_{it}$ are the locate
requests received in shares at a particular time, $t$, for security,
$i$.
\begin{align}
\text{Prob}\left(D_{it}\right) & =\frac{e^{-\lambda_{i}}\left(\lambda_{i}\right)^{D_{it}}}{\left(D_{it}\right)!}\label{eq:3}\\
\text{Locate\;\ Process} & \Longleftrightarrow\text{Poission\;\ Process\;\ with\;\ Arrival\;\ Rate},\;\lambda_{i}\label{eq:4}\\
\text{Alternately},\;D_{it} & \sim\left|N\left(\mu_{D_{i}},\sigma_{D_{i}}^{2}\right)\right|,\;\text{Absolute\;\ Normal\;\ Distribution}\label{eq:5}
\end{align}

Given the number of GBMs the complete system incorporates, a standard
theoretical approach to solving systems involving multiple GBMs (Eq:
\ref{eq:1}) or obtaining a closed form solution, is presently unknown
to the best of our knowledge. In later sections, we provide closed
form solutions to simple scenarios. An alternate approach would be
to estimate the parameters of all the random variables from historical
data and run simulations that would provide the required valuation.
We also provide another technique to use the historical data set directly
in Section (\ref{sec:Lost-in-Estimation:}). Another simplifying assumption
in the numerical results is that the Wiener process governing each
of these variables is independent\footnote{The Wiener process $W_{t}$ is characterised by the following properties
(Baxter \& Rennie 1996):
\begin{enumerate}
\item $W_{0}=0$.
\item $W$ has independent increments: for every $t>0$, the future increments
$W_{t+u}-W_{t}$, $u\geq0$ are independent of the past values $W_{s}$,
$s\leq t$.
\item $W$ has Gaussian increments: $W_{t+u}-W_{t}$ is normally distributed
with mean $0$ and variance $u$, $W_{t+u}-W_{t}\sim{\mathcal{N}}\left(0,u\right)$.
\item $W$ has continuous paths: $W_{t}$ is continuous in $t$.
\end{enumerate}
}. In addition, our baseline models are diffusions without mean-reversion
which we can justify since a term loan contract is unlikely to exceed
one or two years and the variables will not take on excessively large
values in this duration. Non-negative drift rates can grow a variable
to infinity over time, but some of our variables have negative drift
rates as well as we see in the numerical results in Section (\ref{subsec:Sample-Data-Generation}).
Hull (2010) provides an excellent account of using GBMs to model stock
prices and other time series that are always positive. 

It is worth keeping in mind that the intermediary firm or the beneficial
owner will have access to a historical time series of some of the
variables and hence can estimate the actual process for the corresponding
variables. External parties will not know the time series of all the
variables with certainty and hence would need to substitute the unknown
variables with a simulation based process, similar to what we have
used in Section (\ref{subsec:Sample-Data-Generation}). A simplification
is to assume that the variables are independent. Phillips \& Yu (2009)
provide an overview of maximum likelihood and Gaussian methods of
estimating continuous time models used in finance. Campbell et al.
(1998); Lai \& Xing (2008); Cochrane (2009) are other handy resources
on using maximum likelihood estimation (MLE) and generalized method
of moments (GMM). A backward induction based computer program, which
simulates the randomness component of the variables involved, can
calculate the value of the term loans based on the payoff expressions
we will derive. See: Miranda \& Fackler (2002) for a discussion of
using numerical techniques. Chiani, Dardari \& Simon (2003) for approximations
to the error function. Norstad (1999) for a discussion of the log
normal distribution. Gujarati (1995); Hamilton (1994) discuss time
series simplifications and the need for parsimonious models.
\end{doublespace}
\begin{doublespace}

\subsection{\label{subsec:Benchmark-Valuations}Benchmark Valuations}
\end{doublespace}

\begin{doublespace}
The objective of a rational, risk neutral decision maker at the intermediary
would be to maximize profits from the overall loan book by swapping
the supply available to either term loans or regular loans and this
action forms a crucial part of his decision making. With the setting
discussed earlier (Section \ref{subsec:Setting-the-Stage}; Assumption
\ref{Assumption: A-long-term}; Definition \ref{Defn:A-long-term}),
a term loan becomes an American down-and-in binary put option with
the availability process being equivalent to the underlying price
process. The strike and the time to expiration are the quantity and
the time period of the term loan respectively. The volatility governing
this option is the volatility of the availability process. The value
of this option, $\upsilon$, can be expressed as a continuously compounded
annualized rate, $w=\ln\left(\upsilon/H_{i}S_{i0}\right)/T\;\because\;\upsilon=H_{i}S_{i0}e^{wT}$,
over the duration of the loan. Here, $H_{i}$ is the quantity of the
term deal for security $i$ for the entire duration of the loan extending
from $t=0\;to\;t=T.$ $\upsilon,$ is also the valuation of the term
loan for the overall duration. $H_{i}S_{i0}$ represents the notional
amount in monetary terms when the contract is initiated. We first
consider two simple forms of payoffs:
\end{doublespace}
\begin{enumerate}
\begin{doublespace}
\item A single constant payoff, $K_{i}$, if there is a breach: $\left(A_{it}\leq H_{i}\right)\;;\;0\leq t\leq T$.
This applies only the first time the breach happens. It is to be understood
that the contract is complete after the first breach and the payoff
is made (Eq: \ref{eq:6}). The valuation, $\upsilon^{constant}$,
is given by,\textcolor{black}{
\begin{align}
\upsilon^{constant} & =E_{0}\left[K_{i}\boldsymbol{I}\left(A_{it}<H_{i}\right)\right]\;;\;t=\inf\left\{ \;0\leq u\leq T\;\mid\;\left(A_{iu}\leq H_{i}\right)\;\right\} \label{eq:6}
\end{align}
} Here, $\boldsymbol{I}\left(\cdots\right)$ is the indicator function.
$\inf\left\{ \ldots\right\} $ is the infimum or the greatest lower
bound.
\item A single payoff proportional to the time left on the deal, $r\left(T-t\right)/T$
or $K_{i}\left(T-t\right)/T$, if there is a breach: $\left(A_{it}\leq H_{i}\right)\;;\;0\leq t\leq T$.
The proportional payoff could be based on something time varying such
as interest rates or other values as well, denoted here by $r$; it
could also be proportional to a constant value, $K_{i}$. Again, this
is valid only when the first breach happens (Eq: \ref{eq:7}). The
valuation, $\upsilon^{proportional\;time}$, is given by,\textcolor{black}{
\begin{equation}
\upsilon^{proportional\;time}=E_{0}\left[K_{i}\frac{\left(T-t\right)}{T}\boldsymbol{I}\left(A_{it}<H_{i}\right)\right]\;;\;t=\inf\left\{ \;0\leq u\leq T\;\mid\;\left(A_{iu}\leq H_{i}\right)\;\right\} \label{eq:7}
\end{equation}
}
\end{doublespace}
\end{enumerate}
\begin{doublespace}
Carr (1998) gives a semi-explicit approximation for American put option
values in the Black-Scholes model. Garlappi (1996) uses dynamic programming
to value an American put. Zhu (2006) gives an exact and explicit solution
of the Black–Scholes equation for the valuation of American put options.
The closed-form solution is written in the form of a Taylor’s series
expansion that generates a convergent numerical solution if the solution
of the corresponding European option is taken as the initial guess
of the solution series. The optimal exercise boundary is found as
an explicit function of the risk-free interest rate, the volatility
and the time to expiration. Ingersoll (2000) uses digitals as building
blocks, since their payoffs are either on or off, to give accurate
approximations for American options.

\textcolor{black}{Dai \& Kwok (2004) present analytic price formulas
for knock-in American options under the Black-Scholes pricing framework.
The knock-in region and the exercise region of the underlying American
option may intersect with each other, hence the price formulas take
different analytic forms depending on the interaction between the
knock-in region of the down-in feature of the option contract and
the exercise region of the underlying American option. The price function
of a knock-in American option can be expressed in terms of the price
functions of simple barrier options and American options, facilitating
numerical valuation attempts.}

Another interesting technique is the \textcolor{black}{static hedging
}portfolio (SHP) approach. The main idea is to create a static portfolio
of standard European options whose values match the payoff of the
path-dependent options being hedged at expiration and along the boundary.
A SHP is formulated in two different ways. The first approach, proposed
by Carr \& Bowie (1994); Carr, Ellis \& Gupta (1998) is to construct
static positions in a continuum of standard European options of all
strikes, with the maturity date $T$ matching that of the exotic option
(e.g. a barrier option). The second approach, developed by Derman,
Ergener \& Kani (1995), uses a standard European option to match the
boundary at maturity of the exotic option and a continuum of standard
European options of maturities from time $0$ to time $T$ to match
the boundary before maturity of the exotic option, with the strike
equaling the boundary before maturity (e.g. the barrier level of a
knockout option).

\textcolor{black}{In comparison to dynamic hedging, there are three
major advantages of static hedging (Chung \& Shih 2009).} First, static
hedging is considerably cheaper than dynamic hedging when transaction
costs are large. The dynamic hedge portfolio has to be adjusted often,
for example for options with large gamma (such as barrier options),
increasing the transaction costs. Second, it is widely documented
that static hedging is less sensitive to the model risk such as volatility
mis-specification. Third, due to discrete trading, dynamic hedging
may have substantial hedging errors 

\textcolor{black}{Chung \& Shih (2009) further show that the }SHP\textcolor{black}{{}
approach may also serve as a good pricing method for American options
and by its nature, a good hedging method as well. Unlike the use of
numerical methods to pricing American options, one specific advantage
of applying static hedge techniques is that the recalculation of the
American option price in the future is as easy as the valuations of
European options because there is no need to solve the static hedge
portfolio again since the value of the static hedge portfolio is simply
the summation of the European option prices in the portfolio. Hence,
the proposed static hedge approach is especially advantageous when
the European options have closed-form solutions.}

Chung, Shih \& Tsai (2013) extend the SHP approach to price and hedge
American knock-in put options under the Black–Scholes model and the
constant elasticity of variance (CEV) model of Cox (1975; 1996). They
first derive the American knock-in option values on the barrier and
then construct a SHP which matches the knock-in option prices before
maturity at evenly-spaced time points on the barrier. They use standard
European calls to construct SHPs for American up-and-in put options
and standard European puts to construct SHPs for American down-and-in
put options, respectively.

\textcolor{black}{A cash or nothing American binary put option has
a closed form solution, given by the expressions for the Laplace transform
of the distribution of the conditional first passage time of Brownian
motion to a particular level (}Shreve 2004; \textcolor{black}{Azimzadeh
2015; Eq: \ref{eq:8}; \ref{eq:8-A}). The asset-or-nothing case is
a simple scaling (by the strike price) of the cash-or-nothing case.
The price of a cash-or-nothing American binary put with strike $H_{i}>0$,
volatility of availability process $\sigma>0$, drift $\mu\in\Re$,
the set of real numbers and time-to-expiry $T$ is given below. It
is assumed $H_{i}<A_{i0}$, the initial value of availability, since
a binary option is exercised as soon as it is in the money. 
\begin{equation}
\upsilon^{constant}=\frac{K_{i}}{2}e^{a\left(\xi-b\right)}\left\{ 1+\text{sgn}\left(a\right)\text{erf}\left(\frac{bT-a}{\sqrt{2T}}\right)+e^{2ab}\left[1-\text{sgn}\left(a\right)\text{erf}\left(\frac{bT+a}{\sqrt{2T}}\right)\right]\right\} \label{eq:8}
\end{equation}
\begin{eqnarray}
\text{Here, }\text{erf\ensuremath{\left(\right)} is the error function};\text{\text{sgn}\ensuremath{\left(\right)} is the sign function};\label{eq:8-A}\\
a=\frac{1}{\sigma}\log\frac{H_{i}}{A_{i0}};\text{ }\xi=\frac{\mu}{\sigma}-\frac{\sigma}{2};\text{ }b=\sqrt{\xi^{2}+2\mu};\text{ }\xi^{2}+2\mu\geq0;\text{ }\text{i.e. }\ensuremath{b}\text{ is real.}\nonumber 
\end{eqnarray}
}

Since the work of Boyle (1977) provided a proper framework for Monte
Carlo pricing of options, numerous extensions have been done. Longstaff
\& Schwartz (2001) continues to be a popular technique that uses least
squares to estimate the conditional expected payoff to the option-holder
from continuation, making this approach readily applicable in path-dependent
and multi-factor situations where traditional finite difference techniques
cannot be used. Zanger (2018) analyzes the convergence of the Longstaff–Schwartz
algorithm using a single set of independent Monte Carlo sample paths
that is repeatedly reused for all exercise time-steps.

Rogers (2002); Haugh \& Kogan (2004) calculate the lower and upper
bound of American option prices using Monte Carlo simulation by representing
the price as a solution of a properly defined dual minimization problem.
Nadarajah, Margot \& Secomandi (2017) look at least squares Monte
Carlo methods with applications to energy real options.
\end{doublespace}
\begin{doublespace}

\subsection{\label{subsec:Extremely-Exotic-Extensions}Extremely Exotic Extensions}
\end{doublespace}

\begin{doublespace}
The below payoff expressions are more realistic and reflect the full
cost borne by the desk to fulfill the obligations on a term loan.
The American style exercise option is only partly applicable, since
the proper way to look at them is by considering a series of cash
flows being exchanged between the two parties for the duration of
the loan, though any mid-term terminal clauses can be modelled as
American options being exercised.
\end{doublespace}
\begin{enumerate}
\begin{doublespace}
\item A constant or proportional payoff is levied every-time the barrier
is passed in a downward direction giving the valuations, $\upsilon^{constant\;counter}$
or $\upsilon^{proportional\;counter}$ respectively. The indicator
function can be used to count the number of instances when this occurs,
i.e. $\sum_{t=0}^{T}\boldsymbol{I}\left(A_{it}<H_{i}\right)$. When
the availability moves back above the barrier, a reverse cash flow
can be accumulated. An approach from Erdos \& Hunt (1953) derives
results regarding the change of signs of sums of random variables.
This can be used to estimate the number of times the availability
falls below the term loan quantity and hence the number of times a
penalty, $K_{i}$ or $K_{i}\left(T-t\right)$, is incurred.
\end{doublespace}

\begin{doublespace}
Alternately, we can proceed as follows to arrive at the following
result (Proposition \ref{The-indicator-valuation-expression}). Let
the following variables, $\left\{ \text{Up}_{it},\text{Down}_{it}\right\} $
represented by the corresponding functions below (Eq: \ref{eq:9};
\ref{eq:10}), denote the criteria that captures whether availability
is more than the quantity on the term loan for security $i$ or vice
versa and the cash-flows to be made accordingly. When availability
is more than the quantity, we are in the $\text{Up}_{it}$ state.
This means, that the barrier breach that can happen next when we are
in the $\text{Up}_{it}$ state is in the reverse direction, or it
will be in the downward direction. 

\begin{align}
\text{Up State / Down Breach}\equiv\text{Up}_{it} & \equiv\frac{\max\left(A_{it}-H_{i},0\right)}{\left(A_{it}-H_{i}\right)}=\begin{cases}
1\;\text{if}\; & A_{it}\geq H_{i}\\
0 & \text{Otherwise}
\end{cases}\label{eq:9}\\
\text{Down State / Up Breach}\equiv\text{Down}_{it} & \equiv\frac{\max\left(H_{i}-A_{it},0\right)}{\left(H_{i}-A_{it}\right)}=\begin{cases}
1\;\text{if}\; & A_{it}\leq H_{i}\\
0 & \text{Otherwise}
\end{cases}\label{eq:10}
\end{align}

\end{doublespace}
\begin{prop}
\begin{doublespace}
\label{The-indicator-valuation-expression}The valuation expressions
that capture a constant or proportional payoff every-time the barrier
is passed are given by,
\begin{eqnarray*}
\upsilon^{constant\;counter} & = & E_{0}\left[K_{i}\sum_{i=1}^{n}\left(1-\text{Up}_{i0}\right)+K_{i}\sum_{i=1}^{n}\sum_{t=0}^{T-1}\left\lfloor \frac{1+\text{Up}_{it}-\text{Up}_{i,t+1}}{2}\right\rfloor \right.\\
 & - & \left.\tilde{K}_{i}\sum_{i=1}^{n}\sum_{t=0}^{T-1}\left\lfloor \frac{1-\text{Up}_{it}+\text{Up}_{i,t+1}}{2}\right\rfloor \right]
\end{eqnarray*}
\begin{eqnarray*}
\upsilon^{proportional\;counter} & = & E_{0}\left[K_{i}\sum_{i=1}^{n}\left(1-\text{Up}_{i0}\right)+K_{i}\sum_{i=1}^{n}\sum_{t=0}^{T-1}\frac{\left(T-t-1\right)}{T}\left\lfloor \frac{1+\text{Up}_{it}-\text{Up}_{i,t+1}}{2}\right\rfloor \right.\\
 & - & \left.\tilde{K}_{i}\sum_{i=1}^{n}\sum_{t=0}^{T-1}\frac{\left(T-t-1\right)}{T}\left\lfloor \frac{1-\text{Up}_{it}+\text{Up}_{i,t+1}}{2}\right\rfloor \right]
\end{eqnarray*}
Here, $\left\lfloor x\right\rfloor $ is the floor or the greatest
integer function, which gives the largest integer less than $x$.
$E_{0}\left[\cdots\right]$ is the expectation taken at time $t=0$.
$K_{i},\tilde{K}_{i},$ are the constant payoffs on the term deal
if the barrier is breached for security $i$ in the downward and in
the upward direction respectively.
\end{doublespace}
\end{prop}
\begin{proof}
\begin{doublespace}
See Appendix (\ref{subsec:Proof-of-Proposition:The-indicator-valuation-expression}).
\end{doublespace}
\end{proof}
\begin{doublespace}
\item Valuation, $\upsilon^{stock\;holding}$, is based on a payoff equal
to the cost of holding a stock position equal to the extent of shortfall
in any given time period, considered for the entire duration of the
term loan (Eq: \ref{eq:11}). In the first period, we accrue the cost
of buying the amount of short fall and from the subsequent periods
we need to either buy or sell to make up for whether the availability
is higher or lower than the deal amount. After the last time period,
we dispose of any excess stock accumulated. We use the following notations
in some of the formulations below: $\forall$ $a$, $b$ $\in\Re_{+}$,
$\quad$$a^{+}=\max\left\{ a,0\right\} $, $\quad$$a^{-}=\max\left\{ -a,0\right\} $,
and$\quad$ $a\bigwedge b=\min\left\{ a,b\right\} $.
\begin{align}
\upsilon^{stock\;holding} & =E_{0}\left[\sum_{i=1}^{n}\left(H_{i}-A_{i0}\right)^{+}S_{i0}+\sum_{i=1}^{n}\sum_{t=1}^{T}\left[\left(H_{i}-A_{it}\right)^{+}-\left(H_{i}-A_{i,t-1}\right)^{+}\right]S_{it}\right.\label{eq:11}\\
 & \qquad\qquad\left.\phantom{\sum_{i=1}^{n}}-\left(H_{i}-A_{iT}\right)^{+}S_{i,T+1}\right]
\end{align}
This can be simplified to (Eq: \ref{eq:12}),
\begin{equation}
\upsilon^{stock\;holding}=E_{0}\left[\sum_{i=1}^{n}\sum_{t=0}^{T}\left[\left(H_{i}-A_{it}\right)^{+}\right]\left(S_{it}-S_{i,t+1}\right)\right]\label{eq:12}
\end{equation}
In continuous time\footnote{\begin{doublespace}
Brennan (1979) considers contingent claims pricing in discrete time
models. Shieh, Wang \& Yates (1980); Barraud (1981); Sung, Lee \&
Lee (2009) are exhaustive references for methods for the conversion
of discrete to continuous models and vice versa.
\end{doublespace}
} with discounting using the interest rate, $r$, we have (Eq: \ref{eq:13}),
\textcolor{black}{
\begin{equation}
\upsilon^{stock\;holding}=E_{0}\left[\sum_{i=1}^{n}\left\{ \int_{0}^{T}e^{-ru}\left[H_{i}-A_{i}\left(u\right)\right]^{+}dS_{i}\left(u\right)\right\} \right]\label{eq:13}
\end{equation}
} A forward starting loan that starts at time, $\kappa\geq0$, further
than the present time, can be handled by considering a modified summation
(Eq: \ref{eq:14}) below,
\begin{equation}
E_{0}\left[\sum_{i=1}^{n}\sum_{t=\kappa}^{T}\left[\left(H_{i}-A_{it}\right)^{+}\right]\left(S_{it}-S_{i,t+1}\right)\right]\label{eq:14}
\end{equation}
 This has some similarities to Asian options \footnote{\begin{doublespace}
Rogers \& Shi (1995) try a partial-differential-equation, PDE, approach.
Grant, Vora \& Weeks (1997) extend Monte Carlo methods to path dependent
securities and use it to value American Asian options. Milevsky \&
Posner (1998) approximate the finite sum of correlated log-normal
variables required to calculate the payoff of arithmetic Asian options.
Dufresne (2000; 2001) obtains a Laguerre series expansion for both
Asian and reciprocal Asian options. Fusai (2004) prices Asian options
by computing a Laplace transform with respect to time-to-maturity
and a Fourier transform with respect to the logarithm of the strike.
Callegaro, Fiorin \& Grasselli (2019) introduce a pricing methodology
based on the Fourier transform of the asset process. Bormetti et al.
(2018); Jeong et al. (2019) detail other monte carlo based option
pricing methods.
\end{doublespace}
} though we have two sources of uncertainties, from the availability
and from the stock price assuming the interest rate is not stochastic,
and the payoff would be more than just a simple average, though it
involves some form of time summation. Barraquand \& Martineau (1995);
Andersen \& Broadie (2004) provide Monte Carlo based numerical techniques
for multi-dimensional American options. Broadie \& Detemple (1997)
identify optimal exercise strategies and provide valuation formulas
for several types of American options on two or more assets. A powerful
Markov chain method to simulate multivariate distributions is given
in Hastings (1970); Chib \& Greenberg (1995).
\item The examples thus far ignore the borrow rate at which availability
is sourced. A realistic scenario that considers this would add the
payoff from the other scenarios as a spread on top of the borrow rate
giving the valuations, $\upsilon^{borrow\;rate}$. To the cost of
holding a stock position equal to the extent of shortfall, adding
a borrow rate corresponding to the amount on the term loan fulfilled
from the availability, gives a complete formulation for the revenue
that can be expected from this structure. This gives the valuation
of the term loan with three sources of uncertainty (Eq: \ref{eq:15};
\ref{eq:16}). 
\begin{equation}
\upsilon^{borrow\;rate}=E_{0}\left[\sum_{i=1}^{n}\sum_{t=0}^{T}\left\{ \left[\left(H_{i}-A_{it}\right)^{+}\right]\left(S_{it}-S_{i,t+1}\right)+S_{it}Q_{it}\left(H_{i}\text{\ensuremath{\bigwedge}}A_{it}\right)\right\} \right]\label{eq:15}
\end{equation}
\textcolor{black}{
\begin{align}
\upsilon^{borrow\;rate} & =E_{0}\left\{ \sum_{i=1}^{n}\left[\int_{0}^{T}e^{-ru}\left\{ \left[H_{i}-A_{i}\left(u\right)\right]^{+}\right\} dS_{i}\left(u\right)\right.\right.\label{eq:16}\\
 & +\left.\left.\int_{0}^{T}e^{-ru}\left\{ S_{i}\left(u\right)Q_{i}\left(u\right)\left[H_{i}\text{\ensuremath{\bigwedge}}A_{i}\left(u\right)\right]\right\} du\right]\right\} 
\end{align}
The above formulation }(Eq: \ref{eq:15}; \ref{eq:16})\textcolor{black}{{}
treats the borrow rate at any particular point in time, as being applicable
to the entire availability used up for the term deal. A more realistic
scenario can treat even the quantity that applies to the borrow as
being brought in or taken out similar to the way a stock position
is bought or sold at the prevailing prices (}Eq: \ref{eq:17}; \textcolor{black}{\ref{eq:19};
\ref{eq:20}). This is realistic since different borrow amounts are
sourced at different rates. Either form can be used, depending on
the specifics of how the borrow rates are managed by the desk. In
the last period, unlike a stock position, there is no cash-flow from
unwinding the shares borrowed, but for simplicity, we can assume that
the shares borrowed are used for another loan at the prevailing borrow
rate (though the actual proceeds will be higher since the loan rate
will be more than the borrow rate). This gives, 
\begin{align}
\upsilon^{borrow\;rate} & =E_{0}\left[\sum_{i=1}^{n}\sum_{t=0}^{T}\left\{ \left[\left(H_{i}-A_{it}\right)^{+}\right]\left(S_{it}-S_{i,t+1}\right)+\left(H_{i}\text{\ensuremath{\bigwedge}}A_{it}\right)\left[S_{it}Q_{it}-S_{i,t+1}Q_{i,t+1}\right]\right\} \right.\label{eq:17}\\
 & +\left.\sum_{i=1}^{n}\left(H_{i}\text{\ensuremath{\bigwedge}}A_{iT}\right)S_{i,T+1}Q_{i,T+1}\right]\nonumber 
\end{align}
\begin{equation}
\upsilon^{borrow\;rate}\approx E_{0}\left[\sum_{i=1}^{n}\sum_{t=0}^{T}\left\{ \left[\left(H_{i}-A_{it}\right)^{+}\right]\left(S_{it}-S_{i,t+1}\right)+\left(H_{i}\text{\ensuremath{\bigwedge}}A_{it}\right)\left[S_{it}Q_{it}-S_{i,t+1}Q_{i,t+1}\right]\right\} \right]\label{eq:19}
\end{equation}
\begin{align}
\upsilon^{borrow\;rate} & =E_{0}\left\{ \sum_{i=1}^{n}\left[\int_{0}^{T}\left\{ e^{-ru}\left[H_{i}-A_{i}\left(u\right)\right]^{+}\right\} dS_{i}\left(u\right)\right.\right.\label{eq:20}\\
 & +\left.\left.\int_{0}^{T}\left\{ e^{-ru}\left[H_{i}\text{\ensuremath{\bigwedge}}A_{i}\left(u\right)\right]\left[S_{i}\left(u\right)dQ_{i}\left(u\right)+Q_{i}\left(u\right)dS_{i}\left(u\right)+dS_{i}\left(u\right)dQ_{i}\left(u\right)\right]\right\} \right]\right\} \nonumber 
\end{align}
}
\item We could calculate the revenue of the desk by looking at the cost
of sourcing external borrows at the borrow rate, utilizing internal
inventory and exclusive holdings to meet the shorting demands of external
clients at the loan rate (Eq: \ref{eq:22}; \ref{eq:24}). This would
also differ from the previous scenarios by factoring in the different
sources of inventory and the corresponding costs. This assumes that
there is no cost to use internal inventory and a constant daily fee
for the use of the exclusive holdings. Further complications are possible
by including transaction costs for the use of exclusives and funding
rates for the internal inventory. These benchmark revenue figures
can act as a sanity check and provide practical bounds for the term
loan valuation. 
\begin{align}
P= & E_{0}\left[\sum_{i=1}^{n}\sum_{t=0}^{T}\left\{ L_{it}\left(S_{it}R_{it}-S_{i,t+1}R_{i,t+1}\right)-B_{it}\left(S_{it}Q_{it}-S_{i,t+1}Q_{i,t+1}\right)\right\} -fT\right.\label{eq:22}\\
+ & \left.\sum_{i=1}^{n}\left(L_{i,T}S_{i,T+1}R_{i,T+1}-B_{i,T}S_{i,T+1}Q_{i,T+1}\right)\right]\nonumber 
\end{align}
\textcolor{black}{
\begin{align}
P & =E_{0}\left\{ \sum_{i=1}^{n}\left[\int_{0}^{T}\left\{ e^{-ru}L_{i}\left(u\right)\left[S_{i}\left(u\right)dR_{i}\left(u\right)+R_{i}\left(u\right)dS_{i}\left(u\right)+dS_{i}\left(u\right)dR_{i}\left(u\right)\right]\right\} \right.\right.\label{eq:24}\\
 & +\left.\left.\int_{0}^{T}\left\{ e^{-ru}B_{i}\left(u\right)\left[S_{i}\left(u\right)dQ_{i}\left(u\right)+Q_{i}\left(u\right)dS_{i}\left(u\right)+dS_{i}\left(u\right)dQ_{i}\left(u\right)\right]\right\} \right]\right\} -fT\nonumber 
\end{align}
}Here $P$ is the profits of the loan desk from the entire loan book
comprised of $n$ securities over the duration $T$. $f$ is the constant
fee to utilize the exclusive holdings. This is converted from the
payment made for the duration of the exclusive to apply on a daily
basis.
\item Another structure could consider the varying utilizations or shorting
needs of clients, reflecting the scenario when demand or the term
loan quantity is stochastic resulting in the valuation, $\upsilon^{stochastic\;demand}$.
This means the client can be given access to a loan facility and the
utilization of the loan could be changed at the discretion of the
client (Eq: \ref{eq:26}; \ref{eq:27}). Hence, the amount on the
term loan $H_{it}$ or the strike can be made to vary according to
another suitably defined GBM\footnote{\begin{doublespace}
Henderson \& Wojakowski (2002) give a symmetry result between the
floating and fixed-strike Asian options involving a change of numeraire
and time reversal of Brownian motion. Eberlein \& Papapantoleon (2005)
extend this result by considering a general Levy process as the driving
process of the underlying and prove a symmetry result for look-back
options using the same technique. Though in our case, both the strike
and the underlying process vary under different GBMs.
\end{doublespace}
}.
\begin{equation}
\upsilon^{stochastic\;demand}\approx E_{0}\left[\sum_{i=1}^{n}\sum_{t=0}^{T}\left\{ \left[\left(H_{it}-A_{it}\right)^{+}\right]\left(S_{it}-S_{i,t+1}\right)+\left(H_{it}\text{\ensuremath{\bigwedge}}A_{it}\right)\left[S_{it}Q_{it}-S_{i,t+1}Q_{i,t+1}\right]\right\} \right]\label{eq:26}
\end{equation}
\begin{align}
\upsilon^{stochastic\;demand} & =E_{0}\left\{ \sum_{i=1}^{n}\left[\int_{0}^{T}\left\{ e^{-ru}\left[H_{i}\left(u\right)-A_{i}\left(u\right)\right]^{+}\right\} dS_{i}\left(u\right)\right.\right.\label{eq:27}\\
 & +\left.\left.\int_{0}^{T}e^{-ru}\left[H_{i}\left(u\right)\text{\ensuremath{\bigwedge}}A_{i}\left(u\right)\right]\left\{ S_{i}\left(u\right)dQ_{i}\left(u\right)+Q_{i}\left(u\right)dS_{i}\left(u\right)+dS_{i}\left(u\right)dQ_{i}\left(u\right)\right\} \right]\right\} \nonumber 
\end{align}
\textcolor{black}{Considering the stochastic demands of a client,
acts as a natural segue to Section (\ref{sec:Inventory-Management-Application})
on the application of our methodology to inventory management.}
\end{doublespace}
\end{enumerate}
\begin{doublespace}

\section{\label{sec:Inventory-Management-Application}Application to Inventory
Management, Emissions Trading and Insurance Risk Mitigation}
\end{doublespace}

\begin{doublespace}
The methodology we have developed (Sections\textcolor{black}{{} \ref{subsec:Benchmark-Valuations};
\ref{subsec:Extremely-Exotic-Extensions}}) can be useful for supply
chain management, emissions trading and insurance risk mitigation.
We provided numerous examples of how our insights can be useful for
inventory management. We briefly outline how the main ideas can be
used for emissions trading and insurance risk mitigation. We hope
to develop these applications further in subsequent papers. The main
reason for including these examples right now are to illustrate the
wide variety of situations where our techniques be applied.
\end{doublespace}

\subsection{\label{subsec:Inventory-Management}Inventory Management}

\begin{doublespace}
Our methodology can aid in Inventory management by evaluating the
profit and loss over a certain time period. Specifically, the demand
process can be modelled as the underlying price and when it breaches
barriers indicative of the level of supply or inventory, cash-flows
(costs or revenues) can be accumulated over a certain duration. This
can give an estimate of the expected profitability of entering a business,
starting a new venture or maintaining a certain level of inventory.
If the payoffs are used to structure options with other parties, care
needs to be taken to ensure that the retailer or other participants
are not gaming the system and manipulating the demand to meet the
terms of the contract. Khouja (1999); Qin et al. (2011) provide a
comprehensive review of the news-vendor problem including suggestions
for future research. Wang \& Chen (2017) look at option pricing policies
with regards to fresh produce supply chain.

Let $Q_{t}$ represent the supply received (quantity ordered) by the
retailer before time, $t$ with the requirement to sell it before
the next time period, $t+1$ (Eq: \ref{eq:29}). Retailer faces stochastic
demand, $D_{t}$ for the corresponding time period, which can be modelled
using suitable GBMs, perhaps more accurately with jumps. Let the unit
cost to manufacture the product be $c$. The wholesale price at which
manufacturer produces and sells to the retailer be $w$. The salvage
value of any unsold product is $s$ per unit and the stock-out cost
of unsatisfied demand is $r$ per unit. The final price at which retailer
sells is $p$. To avoid unrealistic and trivial cases, we assume that
$0<s<w<p$ , $p>c$ and $0<r$. The profit and loss, $P_{R}$, over
a time period, $t=0$ to $t=T$ can be calculated as,
\begin{equation}
P_{R}=E_{0}\left[\sum_{t=0}^{T}\left\{ \left(D_{t}\text{\ensuremath{\bigwedge}}Q_{t}\right)p-Q_{t}w-\left[\left(D_{t}-Q_{t}\right)^{+}\right]r+\left[\left(D_{t}-Q_{t}\right)^{-}\right]s\right\} \right]\label{eq:29}
\end{equation}
The approach of optimizing the quantity ordered to maximize expected
payoff is well known. Instead or in addition, the retailer could structure
the payoffs when the demand breaches barriers by viewing payoffs as
suitable option contracts. \textcolor{black}{The treatment from Sections
(\ref{subsec:Benchmark-Valuations} ; \ref{subsec:Extremely-Exotic-Extensions})
can apply with some modifications as below}. The inventory management
valuations are denoted similar to the terminology we have used in
Sections \textcolor{black}{(\ref{subsec:Benchmark-Valuations}; \ref{subsec:Extremely-Exotic-Extensions})
but} with an extra suffix $inv$ that denotes that these are for inventory
management, $\upsilon_{inv}^{XXXXX}$.
\end{doublespace}
\begin{enumerate}
\begin{doublespace}
\item If demand falls below a certain threshold, $H$, that is, if there
is a breach $\left(D_{t}\leq H\right)\;;\;0\leq t\leq T$ in a particular
time period within the full duration, the retailer incurs a loss or
payoff $K$. This covers the case wherein a certain quantity is ordered
for all the time periods over a predetermined horizon, perhaps with
contractual commitment. An additional cost (storage, disposal or salvage
value less than wholesale price) is involved when there is excess
inventory on any particular time period within the full horizon of
the contract (Eq: \ref{eq:30}). This applies only the first time
the breach happens and is similar to an American down-and-in binary
put option yielding the below formulation, \textcolor{black}{
\begin{equation}
\upsilon_{inv}^{constant}=E_{0}\left[K\boldsymbol{I}\left(D_{t}<H\right)\right]\;;\;t=\inf\left\{ \;0\leq u\leq T\;\mid\;\left(D_{u}\leq H\right)\;\right\} \label{eq:30}
\end{equation}
}
\item A single payoff proportional to the time left on the deal (Eq: \ref{eq:31}).
This payoff could be before the supply contract can be renegotiated
and could be related to the interest rates. This proportionality is
represented as $r\left(T-t\right)/T$ or $K\left(T-t\right)/T$, if
there is a breach $\left(D_{t}\leq H\right)\;;\;0\leq t\leq T$. Again,
this is valid only when the first breach happens.\textcolor{black}{
\begin{equation}
\upsilon_{inv}^{proportional\;time}=E_{0}\left[K\frac{\left(T-t\right)}{T}\boldsymbol{I}\left(D_{t}<H\right)\right]\;;\;t=\inf\left\{ \;0\leq u\leq T\;\mid\;\left(D_{u}\leq H\right)\;\right\} \label{eq:31}
\end{equation}
The scenarios where demand goes above a certain upper limit giving
rise to costs, }$\tilde{K}$,\textcolor{black}{{} from loss of goodwill
due to shortages etc. can be handled analogously.}
\item Other extensions can cover scenarios where c\textcolor{black}{osts
}$K,\tilde{K},$ \textcolor{black}{are incurred every-time} \textcolor{black}{there
is a demand breach in the downward or upward direction }(Eq: \ref{eq:32};
\ref{eq:33}) \textcolor{black}{respectively,}
\begin{align}
\text{Up State / Down Breach} & \equiv\text{Up}_{t}\equiv\frac{\max\left(D_{t}-H,0\right)}{\left(D_{t}-H\right)}=\begin{cases}
1\;\text{if}\; & D_{t}\geq H\\
0 & \text{Otherwise}
\end{cases}\label{eq:32}\\
\text{Down State / Up Breach}\equiv & \text{Down}_{t}\equiv\frac{\max\left(H-D_{t},0\right)}{\left(H-D_{t}\right)}=\begin{cases}
1\;\text{if}\; & D_{t}\leq H\\
0 & \text{Otherwise}
\end{cases}\label{eq:33}
\end{align}
The valuation expressions that capture a constant (Eq: \ref{eq:34})
or proportional payoff (Eq: \ref{eq:36}) every-time the barrier is
passed are given by,
\end{doublespace}

\begin{doublespace}
\begin{align}
\upsilon_{inv}^{constant\;counter} & =E_{0}\left[K\left(1-\text{Up}_{t}\right)+K\sum_{t=0}^{T-1}\left\lfloor \frac{1+\text{Up}_{t}-\text{Up}_{t+1}}{2}\right\rfloor \right.\label{eq:34}\\
- & \left.\tilde{K}\sum_{t=0}^{T-1}\left\lfloor \frac{1-\text{Up}_{t}+\text{Up}_{t+1}}{2}\right\rfloor \right]\nonumber 
\end{align}
\begin{align}
\upsilon_{inv}^{proportional\;counter} & =E_{0}\left[K\left(1-\text{Up}_{t}\right)+K\sum_{t=0}^{T-1}\frac{\left(T-t-1\right)}{T}\left\lfloor \frac{1+\text{Up}_{t}-\text{Up}_{t+1}}{2}\right\rfloor \right.\label{eq:36}\\
- & \left.\tilde{K}\sum_{t=0}^{T-1}\frac{\left(T-t-1\right)}{T}\left\lfloor \frac{1-\text{Up}_{t}+\text{Up}_{t+1}}{2}\right\rfloor \right]\nonumber 
\end{align}

\end{doublespace}
\begin{doublespace}
\item In the above cases, the costs are time invariant or proportional to
time, without any dependence on the demand process or the level of
inventory. All the cost variables can be set as stochastic processes
with a slight abuse of notation or by assuming some of them hold at
the start and others at end of each time period (Eq: \ref{eq:38};
\ref{eq:39}), giving, 
\begin{equation}
\upsilon_{inv}^{stochastic\;costs}=E_{0}\left[\sum_{t=0}^{T}\left\{ \left(D_{t}\text{\ensuremath{\bigwedge}}Q_{t}\right)p_{t}-Q_{t}w_{t}-\left[\left(D_{t}-Q_{t}\right)^{+}\right]r_{t}+\left[\left(D_{t}-Q_{t}\right)^{-}\right]s_{t}\right\} \right]\label{eq:38}
\end{equation}
\begin{align}
\upsilon_{inv}^{stochastic\;costs} & =E_{0}\left[\int_{0}^{T}e^{-ru}\left\{ \left[D\left(u\right)\text{\ensuremath{\bigwedge}}Q\left(u\right)\right]p\left(u\right)-Q\left(u\right)w\left(u\right)\right.\right.\label{eq:39}\\
- & \left.\left.\left(\left[D\left(u\right)-Q\left(u\right)\right]^{+}\right)r\left(u\right)+\left(\left[D\left(u\right)-Q\left(u\right)\right]^{-}\right)s\left(u\right)\right\} du\vphantom{\int_{0}^{T}\frac{e^{-ru}}{e^{-ru}}}\right]\nonumber 
\end{align}
\end{doublespace}
\end{enumerate}
\begin{doublespace}
Increasing globalization has meant that firms looking for growth might
need to have business partners in other countries, which brings with
it many elements of uncertainty and the necessity of better risk management
tools (Rugman 1976; Mascarenhas 1982; Miller 1992; Reeb, Kwok \& Baek
1998). The techniques in this section can be helpful for firms to
manage contract risks as they negotiate terms with suppliers, distributors,
make investments in different parts of the world and seek to display
socially responsible behavior (Levy 1995; Choi, Lee \& Kim 1999; Celly,
Spekman \& Kamauff 1999; Strike, Gao \& Bansal 2006). 
\end{doublespace}
\begin{doublespace}

\subsection{\label{sec:Applications-to-Emissions} Emissions Trading }
\end{doublespace}

\begin{doublespace}
In recent years, there has been a surge in interest in emissions trading.
Rubin (1996) provides a general treatment of emission trading, banking,
and borrowing in an inter-temporal, continuous-time model using optimal-control
theory. He shows that an efficient equilibrium solution exists that
achieves the least cost, which is the solution attained by a social
planner who knows the cost functions of all firms. Daskalakis, Psychoyios
\& Markellos (2009) study the three main markets for emission allowances
within the European Union Emissions Trading Scheme. They develop an
empirically and theoretically valid framework for the pricing and
hedging of intra-phase and inter-phase futures and options on futures,
respectively (also: Bayer \& Aklin 2020; Gladwin \& Walter 1976).
Wen, Wu \& Gong (2020) analyze the impact of carbon emissions’ environmental
regulation on the stock returns of companies. Narassimhan et al. (2018)
is a comprehensive review of the implementation of emissions trading
systems in various jurisdictions across the globe.

We relate our methodology (Sections\textcolor{black}{{} \ref{subsec:Benchmark-Valuations};
\ref{subsec:Extremely-Exotic-Extensions}}) to emissions trading indirectly
through the demand process. Higher demand, for certain products, will
generate higher levels of emissions due to the corresponding dependency
on manufacturing the required output. Option contracts, based on the
demand process, can be structured depending on whether certain upper
limits of demand will be breached. Hence, options that breach certain
upper limits of demand will indicate when emission thresholds might
be breached. This will indicate whether any payments might need to
be made or emissions contracts can be put in place to handle the outcomes
accordingly. A direct approach can be along the lines of Chaabane,
Ramudhin \& Paquet (2012), who present a model that explicitly considers
environmental costs and can assist decision makers in designing sustainable
supply chains over their entire life cycle. This facilitates the understanding
of optimal supply chain strategies under different environmental policies
for recycling and green house gas emissions reductions.
\end{doublespace}
\begin{doublespace}

\subsection{\label{subsec:Insurance-Risk-Mitigation}Insurance Risk Mitigation }
\end{doublespace}

\begin{doublespace}
Anastasiadis \& Chukova (2012) is a comprehensive literature review
that summarizes the results from a collection of research papers that
relate to modeling insurance claims and the processes associated with
them. Kliger \& Levikson (1998) use the demand for insurance to find
the optimal premium an insurer should charge. Brockett \& Xiaohua
(1997) review the applications of operations research methods in the
insurance industry. Mills (2005; 2012) look at the role insurance
firms need to play in managing climate change risk. Michel-Kerjan
\& Kunreuther (2011); Michaels et al. (1997). suggest that government
programs, scientific institutions and private insurers have to work
together to provide flood and other disaster insurance since the losses
from large catastrophes can deplete the capital pools that reinsurance
companies hold.

The amount of insurance claims any insurer faces in a given period
will be the variable that dictates how much reinsurance coverage is
necessary. The amount of claims will be modeled using our methodology
presented in Sections (\ref{subsec:Extremely-Exotic-Extensions};
\ref{sec:Inventory-Management-Application};\ref{sec:Applications-to-Emissions}).
Higher levels of claims will need higher coverage. Option contracts
can be structured that will depend on the whether certain upper limits
on insurance claims are breached. Payments can be put in place depending
on whether the corresponding options contracts cross any thresholds
or boundaries of claims made in a given time period. In a similar
way, asset and liability management can be done based on how the net
amounts of assets and liabilities change over time. If certain thresholds
are breached, option contracts related to certain monetary amounts
can be in put place. This will ensure that asset and liability mismatches
can be handled with less discrepancies. Since our methodology is built
using many features of binary options, catastrophic events can be
handled depending on whether particular disaster related incidents
occur.
\end{doublespace}
\begin{doublespace}

\section{\label{sec:Lost-in-Estimation:}Lost in Estimation: Option Valuation
using Historical Time Series}
\end{doublespace}

\begin{doublespace}
As a complementary technique to performing a standard simulation based
valuation, we show a way to utilize the rich historical data-set at
the disposal of the intermediary. As it will become clear, this technique
can be used to price many kinds of derivative contracts and has wide
application beyond pricing term loans. Given the complexity and the
number of variables to be estimated, a simple transformation provides
a heuristic to calculate the payoff from the historical time series
of each of the variables. This can then be used as a possible guide
to the calculation of the term loan rate. The initial values of the
stochastic variables in the historical data-set, as well as the barrier
level or strike as necessary, can be scaled by a constant to match
the values of the corresponding variables at the start of the option
contract. This gives us a path of the evolution of the stochastic
process, the GBM in our case. Depending on the amount of historical
data, the historical series is split into partitions and the valuation
is done on each portion and then averaged. 

Calculating the value of the option contract directly ensures that
nothing is lost in translation by first estimating the parameters
of a distribution and then using those to perform the valuation. A
key assumption made is that the GBM process has stationary moments
or that the volatility is not time varying. Clearly this is different
from the change of numeraire technique (Geman, El-Karoui \& Rochet
1995; Benninga, Björk \& Wiener 2002) since a numeraire is a traded
asset and a constant is not, unless interest rates are zero, otherwise
there will be risk-less arbitrage opportunities. An open question
is regarding the convergence of the valuation using the historical
time series (Cowles \& Carlin 1996; Broadie \& Glasserman 1997; Sherman,
Ho \& Dalal 1999; Shapiro \& Homem-de-Mello 2000; Stentoft 2004).
Convergence can be improved by making more partitions, including overlapping
samples of the historical series and applying this method to each
portion. The efficiency of this new approach with overlapping paths,
especially in the presence of jumps in the stochastic process, needs
to compared to the efficacy of standard Monte Carlo variance reduction
techniques (Boyle, Broadie \& Glasserman 1997; Glasserman 2003).

In our case, the scaling constants are chosen such that the starting
value of the historical availability and stock price time series match
the corresponding values at the inception of the contract (Eq: \ref{eq:41};
\ref{eq:42}). $\theta_{ij},\lambda_{ij}$ are the scaling constants
for security $i$ in the partition $j$ for the availability and stock
price time series. $p$ is the number of partitions of the historical
time series for option valuation. $T_{js}$ and $T_{je}$ are the
start and end times of the $j^{th}$ historical time series. The amount
of the term loan or the strike is not scaled in the structures we
have discussed, but might be necessary in other situations. Using
this the historical valuation for the cost of holding a stock proportional
to the extent of shortfall, $\upsilon^{stock\;holding\;historical}$,
is shown below,

\begin{equation}
\upsilon^{stock\;holding\;historical}=\frac{1}{p}\sum_{j=1}^{p}\sum_{i=1}^{n}\sum_{t=-T_{js}}^{-T_{je}}\left[\left(H_{i}-\hat{A}_{it}\right)^{+}\right]\left(\hat{S}_{it}-\hat{S}_{i,t+1}\right)\label{eq:41}
\end{equation}
\begin{equation}
\text{Here, }\theta_{ij}=\frac{A_{i0}}{A_{iT_{js}}};\;\lambda_{ij}=\frac{S_{i0}}{S_{iT_{js}}};\;\hat{A}_{it}=\theta_{ij}A_{it};\;\hat{S}_{it}=\lambda_{ij}S_{it}\label{eq:42}
\end{equation}

All the different valuations structures, 

\textcolor{black}{$\left\{ \upsilon^{constant},\upsilon^{proportional\;time},\upsilon^{constant\;counter},\right.$
$\left.\upsilon^{proportional\;counter},\upsilon^{stock\;holding},\upsilon^{borrow\;rate},\upsilon^{stochastic\;demand}\right\} $,}
can be calculated using the historical approach. This approach also
allows the calculation of the time series of the daily profits that
would accrue to the intermediary firm. The volatility of the daily
profits can be suggestive in terms of how aggressive one should be
in picking one of the various alternative loan structures.

We give an illustration of how different paths can be generated from
a single historical time series and used to value options akin to
a simulation based procedure. Figure (\ref{fig:Full-HistoricalTime-Series})
shows the full time series that is available. This includes the starting
value, the drift and the volatility, which were also chosen randomly
from suitable uniform distributions, since we are simulating this
series. The full series has about 2590 observations, equivalent to
say ten years of historical data. In Figure (\ref{fig:Mutiple-Paths})
we show multiple paths that have been generated by splitting the overall
series into smaller portions and scaling the starting value of each
sub-portion to coincide with the starting value of our underlying
process. In Figure (\ref{fig:No-Overlap}) we split the full series
into ten non-overlapping portions and in Figure (\ref{fig:With-Overlap})
we split the full series into nineteen paths with almost half of each
sub-portion overlapping with one of the other sub-portions. In both
the overlapping and non-overlapping cases, each smaller series has
259 observations. We see that despite the overlapping, we have a rich
set of paths that captures the jumps or movements inherent in the
true process.

\begin{figure}[H]
\subfloat[Full Series]{\includegraphics[width=16.7cm]{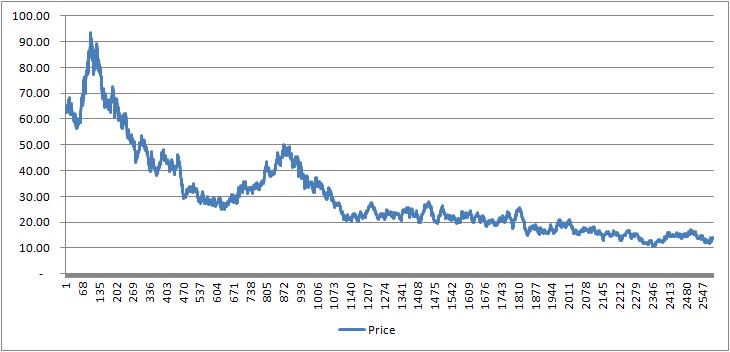}

}

\hfill{}\subfloat[Start Value, Drift and Volatility]{\includegraphics{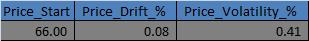}

}\caption{Full Historical Time Series\label{fig:Full-HistoricalTime-Series}}
\end{figure}

\begin{figure}[H]
\subfloat[No Overlap\label{fig:No-Overlap}]{\includegraphics[width=16.7cm]{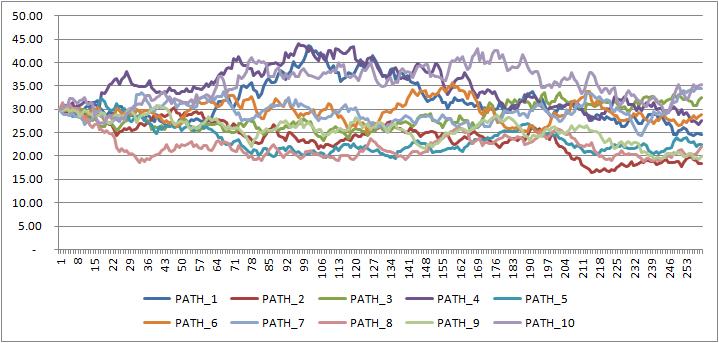}

}

\hfill{}\subfloat[With Overlap\label{fig:With-Overlap}]{\includegraphics[width=16.7cm]{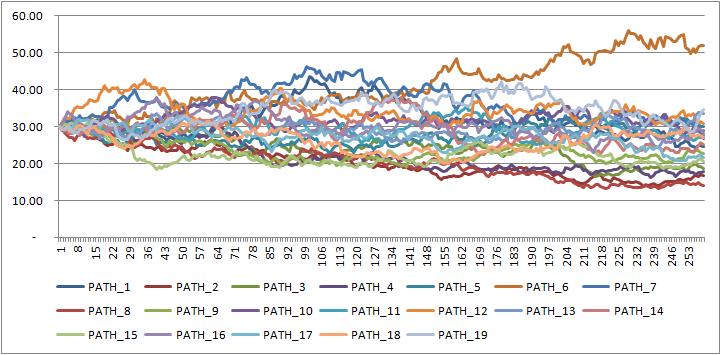}

}\caption{Multiple Paths\label{fig:Mutiple-Paths}}
\end{figure}

\end{doublespace}
\begin{doublespace}

\subsection{\label{subsec:Sample-Data-Generation}Sample Data Generation}
\end{doublespace}

\begin{doublespace}
\textcolor{black}{A typical intermediary can have positions on anywhere
from a few hundred to upwards of a few thousand different securities
and many years of historical data. It is therefore, a good complement
to use the historical time series and calculate the valuation from
the corresponding formulae derived in Sections (\ref{subsec:Benchmark-Valuations};
\ref{subsec:Extremely-Exotic-Extensions}). To demonstrate how this
technique would work, we simulate the historical time series. Also,
the term loan buyer (as opposed to the intermediary, who would be
the writer of the contract) is unlikely to have access to the full
historical time series (but might have the time series of loan rates
and availability) and hence could simulate the variables for which
the historical data is absent as shown in this section to come up
with a valuation. We create one hundred different hypothetical securities
with the same starting value, volatility and drift and come up with
one hundred different paths and hence the different hypothetical time
series of all the variables involved (Price, Availability, Quantity
Borrowed, Exclusive Holding, Inventory Level, Loan Rate, Alternate
Loan Rate) by sampling from suitable log normal distributions. It
is worth noting that the starting value, mean and standard deviation
of the time series are themselves simulations from other appropriately
chosen uniform distributions (Figure \ref{fig:Simulation-Seed}).
The locate process can be modelled as a Poisson distribution with
appropriately chosen units. It is simpler to consider it as the absolute
value of a normal distribution. The mean and standard deviation of
the locate distribution for each security are chosen from another
appropriately chosen uniform distribution.}

The simulation seed is chosen so that the drift and volatility we
get for the variables (mean and standard deviation for the locate
process) are similar to what would be observed in practice. For example
in Figure (\ref{fig:Simulation-Seed}), the price and rate volatility
are lower than the volatilities of the borrow and other quantities,
which tend to be much higher; the range of the drift for the quantities
is also higher as compared to the drift range of prices and rates.
 This ensures that we are keeping it as close to a realistic setting
as possible, without having access to an actual historical time series.
The volatility and drift of the variables for each security are shown
in Figure (\ref{fig:Simulation-Sample-Distributions}). The length
of the simulated time series is one year or 252 trading days for each
security. A sample of the time series of the variables generated using
the simulated drift and volatility parameters is shown in Figure (\ref{fig:Simulation-Sample-Time}).
The full time series shown below is available upon request.

\begin{figure}[H]
\includegraphics[width=7cm]{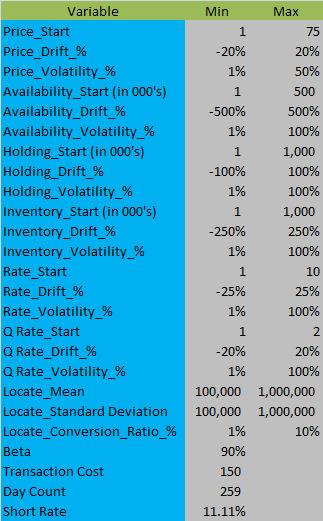}

\caption{Simulation Seed\label{fig:Simulation-Seed}}
\end{figure}
\begin{figure}[H]
\includegraphics[width=15cm]{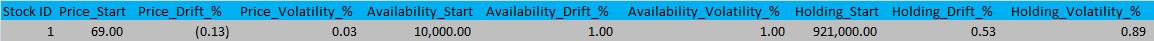}\medskip{}

\includegraphics[width=16.7cm]{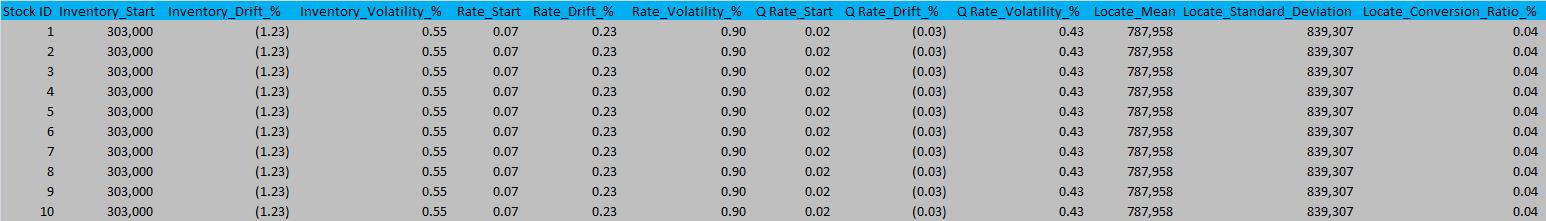}

\caption{Simulation Sample Distributions\label{fig:Simulation-Sample-Distributions}}
\end{figure}
\begin{figure}[H]
\includegraphics[width=8cm]{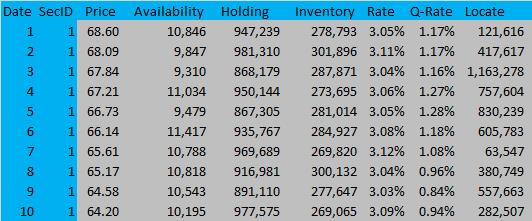}\enskip{}\includegraphics[width=8cm]{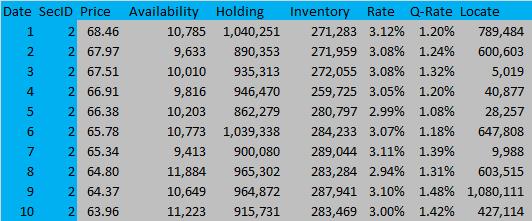}

\caption{Simulation Sample Time Series\label{fig:Simulation-Sample-Time}}
\end{figure}

\end{doublespace}
\begin{doublespace}

\section{\label{sec:Simulate-and-Accumulate}Simulate and Accumulate}
\end{doublespace}
\begin{doublespace}

\subsection{\label{subsec:Closed-Form-Benchmark}Closed Form Benchmark Valuation}
\end{doublespace}

\begin{doublespace}
We first show a grid with the results for the American binary put
closed form valuation using Equation (\ref{eq:8}) in Figure (\ref{fig:American-Binary-Put})
in currency terms. As expected, the results seem to confirm well known
properties such as; longer expiration date, higher volatility, starting
value and strike (term loan amount) being closer, lower drift away
from the strike; lead to higher valuations. For example: when the
availability starting value, availability drift, availability volatility,
term loan amount, payment when barrier is breached and maturity date
are 10000, 100\%, 100\%, 5000, 1000 and 1 year respectively, the valuation
is 225.83 in terms of currency units. This is our baseline scenario
and we vary one variable at a time and keep the other variables at
these baseline values. When the maturity increases to 2 years the
valuation is 246.16. When the availability volatility increases across
10\%, 50\%, 100\% and 150\% the valuation increases from 0, 3.43,
225.83 and 508.56 units respectively. When the availability staring
value decreases to 8000 and 6000 the valuation increases to 370.51
and 694.81 respectively. When the availability drift changes across
50\%, 100\% and 200\% the valuation varies across 400.69, 225.83 and
61.70 respectively. This suggests that the valuation is most sensitive
to the volatility and starting value (with reference to the strike)
of the underlying process, though further analysis needs to be carried
out to find out the exact sensitivity of the valuation to changes
in the different parameters.

\begin{figure}[H]
\includegraphics[width=16.7cm]{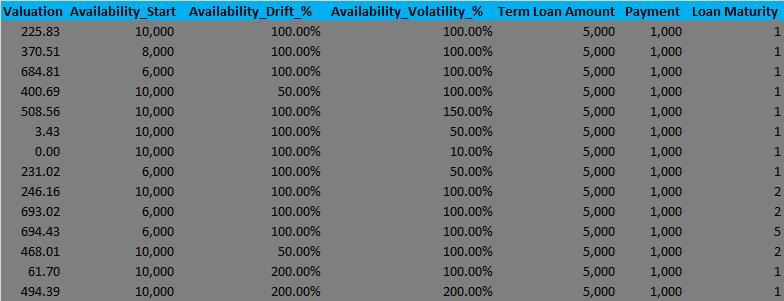}

\caption{American Binary Put Closed Form Valuations\label{fig:American-Binary-Put}}
\end{figure}

\end{doublespace}
\begin{doublespace}

\subsection{\label{subsec:Valuation-Matrix}Valuation Matrix}
\end{doublespace}

\begin{doublespace}
As noted earlier, given the complexity of the system and the number
of random variables involved, the computational infrastructure required
to value term loans would require a Monte-Carlo engine and can be
tremendous; though, such a framework is readily accessible since most
existing intermediaries have this setup as part of their derivative
desks. We build a simple Monte-Carlo framework in R (Paradis 2002;
Matloff 2011; Venables \& Smith 2016). To try different scenarios,
each of the variables (availability drift / start value / volatility,
interest rate, payment up / down and expiration date) are increased
and decreased by 10\% for five steps to get a maximum increase (decrease)
of 150\% (50\%) of the initial values of the corresponding variables.
We also repeat each of the variable value change scenarios with different
number of iterations of the Monte-Carlo sampling, ranging from 5000
to 50,000 with an increase of 5000 iterations each time. This results
in a total of 611 different scenarios for each valuation.

Figures (\ref{fig:Simulation-Results-with-Iteration}; \ref{fig:Simulation-Results-with-Availability-Volatility};
\ref{fig:Simulation-Results-with-Maturity-Date}) show the valuations
when the number of iterations, the availability volatility and the
maturity date are varied. The columns: Constant, Const\_Disc, Proportional,
Prop\_Disc, Up\_Down, Up\_Down\_Prop, Holding, Borrow denote the valuations
corresponding to: a constant payoff, a constant payoff discounted
back to the starting time, a payoff proportional to the time left
on the contract, a payoff proportional to the time left on the contract
discounted back to the starting time, cash flow accruals when there
is a breach in the upward or downward directions, cash flow accruals
proportional to the time left on the contract when there is a breach
in the upward or downward directions, the cost of holding stock positions
to counter any drop in the availability below the contract amount,
and when the borrow rates on the securities are considered in the
valuation. Figure (\ref{fig:Values-of-Other-Variables}) shows the
values of the other variables for the sample valuations shown. These
values for the other variables, act as the baseline, are varied one
variable at a time to create the complete valuation matrix.

As seen in Figure (\ref{fig:Simulation-Results-with-Availability-Volatility})
when the availability volatility increases, across the range 50\%
to 150\% in steps of 10\%, the valuation, corresponding to a payoff
proportional to the time left on the contract discounted back to the
starting time, increases across 0.0 to 560.4. The other valuations
show similar increases. As seen in Figure (\ref{fig:Simulation-Results-with-Maturity-Date})
when the maturity increases, from 0.5 to 0.9 in steps of 0.1 and from
2 to 6 in steps of 1, the valuation increases from 91.03 to 165.91
and from 254.6 to 355.75 units respectively. When we compare the constant
payoff discounted back to the starting time and the payoff proportional
to the time left on the contract discounted back to the starting time,
we see that the proportional payoff is lower across all maturities
and availability volatility levels. But as the maturity dates or the
availability volatility increases, the difference between the two
valuations decreases (we check this by taking the ratio of the difference
between the two valuations and the proportional payoff valuation).
This suggests that it is less risky to offer the contract with a proportional
payoff compared to the constant payoff. For shorter maturities and
when availability has lower volatility the difference in valuations
is greater (hence the risk is greater) and the payoff proportional
to the time left on the contract is to be preferred.

We have utilized the models developed here to price long term securities
lending loans using industry standard Monte-Carlo pricing engines.
As compared to the sample model, illustrated with simulated data in
this paper, we get somewhat similar results on a heavy duty option
pricing engine with actual loan rates and availability inputs. Despite
the simplicity of our calculation engine, we build some intelligence
such that it can pick up and resume the computations from the last
point where it was stopped. This is useful since the total computer
time required to perform these valuations was around 56 hours and
any errors that could abruptly stop the calculations would mean having
to repeat from scratch, an extremely time consuming process. The full
matrix of results is available upon request. 

\begin{figure}[H]
\includegraphics[width=16.7cm]{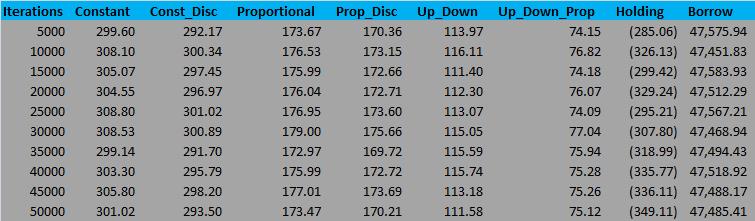}

\caption{Simulation Results with Iteration Changes\label{fig:Simulation-Results-with-Iteration}}
\end{figure}

\begin{figure}[H]
\includegraphics[width=16.7cm]{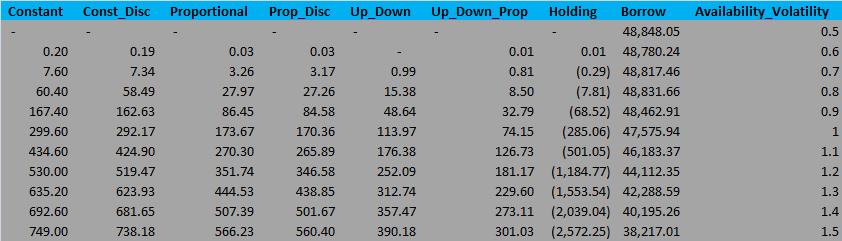}

\caption{Simulation Results with Availability Volatility Changes\label{fig:Simulation-Results-with-Availability-Volatility}}
\end{figure}

\begin{figure}[H]
\includegraphics[width=16.7cm]{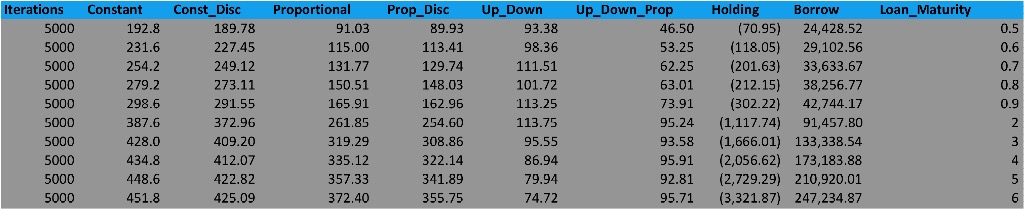}

\caption{Term Loan Valuation, Simulation Results with Maturity Date Changes\label{fig:Simulation-Results-with-Maturity-Date}}
\end{figure}
\begin{figure}[H]
\includegraphics[width=16.7cm]{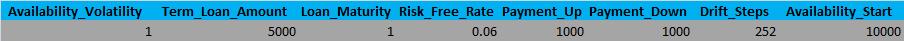}\medskip{}
\includegraphics[width=16.7cm]{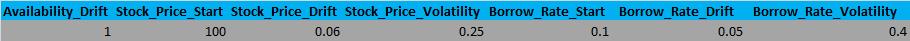}

\caption{Values of other Variables\label{fig:Values-of-Other-Variables}}
\end{figure}

\end{doublespace}
\begin{doublespace}

\section{\label{sec:Improvements-to-the}Improvements to the Model}
\end{doublespace}

\begin{doublespace}
Numerous improvements to the models are possible. Cobb et al. (2007)
derive approximate distributions for the sum of log normal distributions
which highlight that we can estimate the log normal parameters from
the time series of the individual components and hence get the mean
and variance of the availability. Xu (2006) attempts to extend the
complete market option pricing theory to incomplete markets, which
is a more accurate picture of securities lending, since we do not
have assets available for every state of the world. 

An important extension can be to introduce jumps in the stochastic
processes (Merton 1976; Kou 2002). This is seen in stock prices\footnote{\begin{doublespace}
Yan (2011) demonstrates that expected stock return should be monotonically
decreasing in average stock jump size, which can be proxied by the
slope of option implied volatility smile.
\end{doublespace}
} to a certain extent and to a greater extent in the borrow, exclusive
holdings, availability and inventory processes.  Another improvement
would be to use stochastic volatility for the stock prices. The SABR
model, a stochastic volatility model in which the asset price and
volatility are correlated, is discussed in Hagan et al. (2002); Hagan,
Lesniewski \& Woodward (2015). This was developed to overcome the
deficiencies of local volatility models, (Dupire 1994; 1997; Derman
\& Kani 1994; 1998). Gatheral (2011) is an excellent reference for
modeling volatility surfaces. West (2005) considers the SABR model
for illiquid securities, which can be applied to the other processes
in our system.

We have assumed that the stochastic processes governing the stocks,
the loan rates, the number of shares available internally, as well
as the number of shares available for loan from other external borrowers
are all independent. A more realistic assumption of positive correlations
between some of these processes would make the results more realistic
and appealing. Such a modification can be incorporated into our framework
but would require non trivial alternations to the equations and numerical
computations (Heston 1993; Oksendal 2013). Specifically we have assumed
that the process for our most crucial variable, the total availability,
$A_{it}$, expressed as the sum of the other GBMs, is also a GBM.
Even without the presence of correlation between the components, the
process for the total availability might need to modelled differently
(Dufresne 2004; Lo 2012; Nie \& Chen 2007). The main aspect we would
like to emphasize is that we have preferred simpler models to convey
the insights. Incorporating more complex assumptions and models is
a relatively straightforward matter depending on the needs of the
situation and the efforts that can be expended towards the expected
outcomes. That being said, simpler models are to be preferred since
it is easier to follow the intuition and also to see how the results
change based on the inputs and the structure of the variable dependencies.

A key variable we have kept constant is the interest rate. There is
a huge literate on modeling the term structure of short rates (Vasicek
1977; Cox, Ingersoll \& Ross 1985; Ho \& Lee 1986; Hull \& White 1990,
1993; Litterman \& Scheinkman 1991; Heath, Jarrow \& Morton 1992;
Ritchken \& Sankarasubramanian 1995; Brace \& Musiela 1997; Moraleda
\& Vorst 1997; Clewlow \& Strickland 1998). Stock price volatility
surfaces and interest rate processes are part of the standard setup
at most derivative trading desks and can be readily used where available.

The availability process and other processes that follow a GBM but
represent integer valued variables such as number of shares need to
be rounded down to the nearest integer in each time period. For simplicity,
we have ignored this since this would not affect the valuations significantly.
A longer historical time series will help get better estimates for
the volatility of the availability since this is a key factor governing
the uncertainty in the term loan rates. In the numerical results,
we hope to extend the valuation scenarios to include stock price and
loan rate related changes.

A significant amount of analytical work can be done regarding the
use of historical series to get multiple paths for the stochastic
process to perform the valuation. Generally, valuations are done under
the risk-neutral measure to avoid arbitrage, but historical time series
are from the physical measure or the probabilities and associated
values that could actually happen in the real world. Therefore, pricing
based on historical time series may not be arbitrage free. Another
issue that should be considered is that using non-overlapping portions
to price the same instrument creates dependence among the non-overlapping
portions (connects price series from two far away dates, which could
have different variances and other properties due to structural changes),
which could make the price estimator more variable. Among the various
time series properties that could be studied, we could examine any
redundancy created by different levels of the extent of overlap by
comparing it to the no overlap scenario. These checks should gauge
whether the variance of the estimation increases while using the historical
time series and if so try to establish under what conditions it would
yield better results than a standard Monte Carlo simulation.

We have provided some hints regarding the sensitivity of the valuations
to changes in the parameters values and this needs to be investigated
further. We have briefly mentioned the possibility of hedging the
term loan using correlated assets. \textcolor{black}{A key open question
is to decide which of the valuations to use when writing a term loan
when there are no contractual differences to the client but the valuation
simply uses a more complex scenario. This aspect will require views
on how the loan rates might evolve and which securities in the exclusive
pool will stay special or might become special, and hence can be used
to pick either a more aggressive or a less aggressive valuation. In
a subsequent paper, we will look at how we can systematically try
and establish expectations on loan rates and which securities might
become harder to borrow and hence provide higher profit margins on
the loans. Efforts to gauge future demand can be aided by looking
at the locate requests seen by the desk. More locate requests could
indicate a future strain on the availability process keeping in mind
that the locate process is noisy with low locate to loan conversation
ratios. The locate conversion ratio can also be the result of profit
maximization when the Knapsack algorithm }(Martello \& Toth 1987)\textcolor{black}{{}
is used to allocate the locates.}
\end{doublespace}
\begin{doublespace}

\section{\label{sec:Conclusion}Conclusion}
\end{doublespace}

\begin{doublespace}
The paper presents a theoretical foundation supplemented with numerical
results for a largely unexplored financial business. We have developed
models to price long term loans in the securities lending business.
These longer horizon deals can be viewed as contracts with options
embedded in them and can be priced using established methods from
derivatives theory. We have derived the payoff functions for many
structures that can act as term loans. We have provided numerous references
to further explore techniques for solving these and derived the closed
form solutions where such a formulation exists and in situations where
approximations and numerical solutions would be required we have provided
those. The results from the simulation confirm the complexity inherent
in the system, but point out that the techniques we have used can
be a practical tool for the participants (financial intermediaries
and investment firms) to maximize their profits and to improve the
efficiency of the asset management process.

This methodology could also be used to calculate the profit and loss
for inventory management, assuming that the demand process is like
the availability and a payoff or cost occurs every time a barrier
is breached. Similarly it could be potentially useful for dealing
with many forms of uncertainty related to other financial instruments
and even non-financial commodities. This has significant implications
for businesses of all sizes and can become an important risk management
tool in scenarios where there are major risks to the supplier or distribution
channels, such as for multinational firms.

We have developed a heuristic that can mitigate the loss of information
that sets in, when parameters are estimated first and then the valuation
is performed, by directly calculating the valuation using the historical
time series. By first doing an estimation and then using the estimates
for final calculations, we are potentially introducing multiple levels
of errors. Direct use of the historical time series can lead to better
usage of models since the results are more accurate leading to greater
financial stability in the overall system.
\end{doublespace}
\begin{doublespace}

\section{References}
\end{doublespace}
\begin{enumerate}
\begin{doublespace}
\item Adner, R., \& Levinthal, D. A. (2004). What is not a real option:
Considering boundaries for the application of real options to business
strategy. Academy of management review, 29(1), 74
\item Alexander, C., \& Venkatramanan, A. (2012). Analytic Approximations
for Multi‐Asset Option Pricing. Mathematical Finance, 22(4), 667-689.
\item Andersen, L., \& Broadie, M. (2004). Primal-dual simulation algorithm
for pricing multidimensional American options. Management Science,
50(9), 1222-1234.
\item Anastasiadis, S., \& Chukova, S. (2012). Multivariate insurance models:
an overview. Insurance: Mathematics and Economics, 51(1), 222-227. 
\item Azimzadeh, P. (2015). Closed-form expressions for perpetual and finite-maturity
American binary options. See: \href{http://parsiad.ca/blog/2015/pricing-american-binary-options/}{Closed Form Expressions for Perpetual and Finite American Binary Options};
\href{https://en.wikipedia.org/wiki/Binary_option\#American_style}{Closed Form Expressions for American Binary Options, Wikipedia Link}
\item Baklanova, V., Copeland, A., \& McCaughrin, R. (2015). Reference Guide
to US Repo and Securities Lending Markets (No. 740).
\item Baldi, F., \& Trigeorgis, L. (2015). Toward a Real Options Theory
of Strategic Human Resource Management. In Academy of Management Proceedings
(Vol. 2015, No. 1, p. 14862). Academy of Management.
\item Barraquand, J., \& Martineau, D. (1995). Numerical valuation of high
dimensional multivariate American securities. Journal of financial
and quantitative analysis, 30(03), 383-405.
\item Barraud, A. Y. (1981). More on the conversion problem of discrete—continuous
models. Applied Mathematical Modelling, 5(6), 414-416.
\item Battalio, R., \& Schultz, P. (2006). Options and the bubble. The Journal
of Finance, 61(5), 2071-2102.
\item Battalio, R., \& Schultz, P. (2011). Regulatory uncertainty and market
liquidity: The 2008 short sale ban's impact on equity option markets.
The Journal of Finance, 66(6), 2013-2053.
\end{doublespace}
\item Baxter, M., \& Rennie, A. J. (1996). Financial calculus: an introduction
to derivative pricing. Cambridge university press.
\item Bayer, P., \& Aklin, M. (2020). The European Union emissions trading
system reduced CO2 emissions despite low prices. Proceedings of the
National Academy of Sciences, 117(16), 8804-8812.
\begin{doublespace}
\item Belderbos, R., \& Zou, J. (2009). Real options and foreign affiliate
divestments: A portfolio perspective. Journal of International Business
Studies, 40(4), 600-620.
\item Benninga, S., Björk, T., \& Wiener, Z. (2002). On the use of numeraires
in option pricing. The Journal of Derivatives, 10(2), 43-58.
\item Black, F., \& Scholes, M. (1973). The Pricing of Options and Corporate
Liabilities. The Journal of Political Economy, 81(3), 637-654.
\item Bollen, N. P. (1999). Real options and product life cycles. Management
Science, 45(5), 670-684.
\item Boomsma, T. K., Meade, N., \& Fleten, S. E. (2012). Renewable energy
investments under different support schemes: A real options approach.
European Journal of Operational Research, 220(1), 225-237.
\end{doublespace}
\item Bormetti, G., Callegaro, G., Livieri, G., \& Pallavicini, A. (2018).
A backward Monte Carlo approach to exotic option pricing. European
Journal of Applied Mathematics, 29(1), 146-187. 
\begin{doublespace}
\item Boyle, P. P. (1977). Options: A monte carlo approach. Journal of financial
economics, 4(3), 323-338.
\item Boyle, P., Broadie, M., \& Glasserman, P. (1997). Monte Carlo methods
for security pricing. Journal of economic dynamics and control, 21(8),
1267-1321.
\item Brace, A., \& Musiela, M. (1997). The market model of interest rate
dynamics. Mathematical finance, 7(2), 127-155.
\item Brennan, M. J. (1979). The pricing of contingent claims in discrete
time models. The journal of finance, 34(1), 53-68.
\item Broadie, M., \& Detemple, J. (1997). The valuation of American options
on multiple assets. Mathematical Finance, 7(3), 241-286.
\item Broadie, M., \& Glasserman, P. (1997). Pricing American-style securities
using simulation. Journal of Economic Dynamics and Control, 21(8),
1323-1352.
\item Brockett, P., \& Xiaohua, X. (1997). Operations Research In Insurance:
A Review. Insurance Mathematics and Economics, 2(19), 154.
\item Brockman, P., \& Turtle, H. J. (2003). A barrier option framework
for corporate security valuation. Journal of Financial Economics,
67(3), 511-529.
\end{doublespace}
\item Callegaro, G., Fiorin, L., \& Grasselli, M. (2019). Quantization meets
Fourier: a new technology for pricing options. Annals of Operations
Research, 282(1), 59-86.
\begin{doublespace}
\item Campbell, J. Y., Lo, A. W., MacKinlay, A. C., \& Whitelaw, R. F. (1998).
The econometrics of financial markets. Macroeconomic Dynamics, 2(04),
559-562.
\item Carlsson, C., \& Fullér, R. (2001). On possibilistic mean value and
variance of fuzzy numbers. Fuzzy sets and systems, 122(2), 315-326.
\item Carr, P. (1995). Two extensions to barrier option valuation. Applied
Mathematical Finance, 2(3), 173-209.
\item Carr, P., \& Bowie, J. (1994). Static simplicity. Risk, 7(8), 45-50.
\item Carr, P. (1998). Randomization and the American put. Review of Financial
Studies, 11(3), 597-626.
\item Carr, P., Ellis, K., \& Gupta, V. (1998). Static hedging of exotic
options. The Journal of Finance, 53(3), 1165-1190.
\item Celly, K. S., Spekman, R. E., \& Kamauff, J. W. (1999). Technological
uncertainty, buyer preferences and supplier assurances: an examination
of Pacific Rim purchasing arrangements. Journal of International Business
Studies, 30(2), 297-316.
\item Chaabane, A., Ramudhin, A., \& Paquet, M. (2012). Design of sustainable
supply chains under the emission trading scheme. International Journal
of Production Economics, 135(1), 37-49.
\item Chen, C. Y., Wang, H. C., \& Wang, J. Y. (2015). The valuation of
forward-start rainbow options. Review of Derivatives Research, 18(2),
145-188.
\item Cheng, T. T. (1949). The normal approximation to the Poisson distribution
and a proof of a conjecture of Ramanujan. Bulletin of the American
Mathematical Society, 55(4), 396-401.
\end{doublespace}
\item Chi, T., Li, J., Trigeorgis, L. G., \& Tsekrekos, A. E. (2019). Real
options theory in international business. Journal of International
Business Studies, 50(4), 525-553.
\begin{doublespace}
\item Chib, S., \& Greenberg, E. (1995). Understanding the metropolis-hastings
algorithm. The American Statistician, 49(4), 327-335.
\item Chiani, M., Dardari, D., \& Simon, M. K. (2003). New exponential bounds
and approximations for the computation of error probability in fading
channels. Wireless Communications, IEEE Transactions on, 2(4), 840-845.
\item Choi, C. J., Lee, S. H., \& Kim, J. B. (1999). A note on countertrade:
contractual uncertainty and transaction governance in emerging economies.
Journal of International Business Studies, 30(1), 189-201.
\item Chung, S. L., \& Shih, P. T. (2009). Static hedging and pricing American
options. Journal of Banking \& Finance, 33(11), 2140-2149.
\item Chung, S. L., Shih, P. T., \& Tsai, W. C. (2013). Static hedging and
pricing American knock-in put options. Journal of Banking \& Finance,
37(1), 191-205.
\item Clewlow, L., \& Strickland, C. (1998). Pricing interest rate exotics
by Monte Carlo simulation. Monte Carlo: Methodologies and Applications
for Pricing and Risk Management.
\item Cobb, B. R., Rumi, R., \& Salmerón, A. (2012). Approximating the Distribution
of a Sum of Log-normal Random Variables. Statistics and Computing,
16(3), 293-308.
\item Cochrane, J. H. (2009). Asset Pricing:(Revised Edition). Princeton
university press.
\item Cowles, M. K., \& Carlin, B. P. (1996). Markov chain Monte Carlo convergence
diagnostics: a comparative review. Journal of the American Statistical
Association, 91(434), 883-904.
\item Cox, J. C. (1975). Notes on option pricing I: Constant elasticity
of variance diffusions. Unpublished note, Stanford University, Graduate
School of Business.
\item Cox, J. C., Ingersoll Jr, J. E., \& Ross, S. A. (1985). A theory of
the term structure of interest rates. Econometrica: Journal of the
Econometric Society, 385-407.
\item Cox, J. C. (1996). The constant elasticity of variance option pricing
model. The Journal of Portfolio Management, 23(5), 15-17.
\item Cuypers, I. R., \& Martin, X. (2010). What makes and what does not
make a real option? A study of equity shares in international joint
ventures. Journal of International Business Studies, 41(1), 47-69.
\item Dai, M., \& Kwok, Y. K. (2004). Knock-in American options. Journal
of Futures Markets, 24(2), 179-192.
\item Daskalakis, G., Psychoyios, D., \& Markellos, R. N. (2009). Modeling
CO2 emission allowance prices and derivatives: evidence from the European
trading scheme. Journal of Banking \& Finance, 33(7), 1230-1241.
\item D’avolio, G. (2002). The market for borrowing stock. Journal of financial
economics, 66(2), 271-306.
\item Derman, E., \& Kani, I (1994), Riding on a smile, Risk, Feb. pp. 32–39.
E. Derman and I. Kani, (1998), Stochastic implied trees: Arbitrage
pricing
\item Derman, E., Ergener, D., \& Kani, I. (1995). Static options replication.
The Journal of Derivatives, 2(4), 78-95.
\item Derman, E., \& Kani, I. (1998). Stochastic implied trees: Arbitrage
pricing with stochastic term and strike structure of volatility. International
journal of theoretical and applied finance, 1(01), 61-110.
\item Duffie, D., Garleanu, N., \& Pedersen, L. H. (2002). Securities lending,
shorting, and pricing. Journal of Financial Economics, 66(2), 307-339.
\item Dufresne, D. (2000). Laguerre series for Asian and other options.
Mathematical Finance, 10(4), 407-428.
\item Dufresne, D. (2001). The integral of geometric Brownian motion. Advances
in Applied Probability, 33(1), 223-241.
\item Dufresne, D. (2004). The log-normal approximation in financial and
other computations. Advances in Applied Probability, 36(3), 747-773.
\item Dupire, B. (1994). Pricing with a smile. Risk, 7(1), 18-20.
\item Dupire, B. (1997). Pricing and hedging with smiles (pp. 103-112).
Mathematics of derivative securities. Dempster and Pliska eds., Cambridge
Uni. Press.
\item Eberlein, E., \& Papapantoleon, A. (2005). Equivalence of floating
and fixed strike Asian and lookback options. Stochastic Processes
and their Applications, 115(1), 31-40.
\item Erdos, P., \& Hunt, G. A. (1953). Changes of sign of sums of random
variables. Pacific J. Math, 3, 673-687.
\item Evans, R. B., Geczy, C. C., Musto, D. K., \& Reed, A. V. (2009). Failure
is an option: Impediments to short selling and options prices. Review
of Financial Studies, 22(5), 1955-1980.
\item Fernandes, B., Cunha, J., \& Ferreira, P. (2011). The use of real
options approach in energy sector investments. Renewable and Sustainable
Energy Reviews, 15(9), 4491-4497.
\item Froot, K. A. (1995). Hedging portfolios with real assets. The Journal
of Portfolio Management, 21(4), 60-77.
\item Fusai, G. (2004). Pricing Asian options via Fourier and Laplace transforms.
Journal of Computational Finance, 7(3), 87-106.
\item Garlappi, L. (1996). Valuation of the American Put Option: A Dynamic
Programming Approach.
\item Gatheral, J. (2011). The volatility surface: a practitioner's guide
(Vol. 357). John Wiley \& Sons.
\item Geman, H., El Karoui, N., \& Rochet, J. C. (1995). Changes of numeraire,
changes of probability measure and option pricing. Journal of Applied
probability, 32(2), 443-458.
\item Gladwin, T. N., \& Walter, I. (1976). Multinational enterprise, social
responsiveness, and pollution control. Journal of International Business
Studies, 7(2), 57-74.
\item Glasserman, P. (2003). Monte Carlo methods in financial engineering
(Vol. 53). Springer Science \& Business Media.
\item Grant, D., Vora, G., \& Weeks, D. (1997). Path-dependent options:
Extending the Monte Carlo simulation approach. Management Science,
43(11), 1589-1602.
\item Grullon, G., Lyandres, E., \& Zhdanov, A. (2012). Real options, volatility,
and stock returns. The Journal of Finance, 67(4), 1499-1537.
\item Gujarati, D. N. (1995). Basic econometrics, 3rd. International Edition.
\item Hagan, P., Lesniewski, A., \& Woodward, D. (2015). Probability distribution
in the SABR model of stochastic volatility. In Large deviations and
asymptotic methods in finance (pp. 1-35). Springer International Publishing.
\item Hagan, P. S., Kumar, D., Lesniewski, A. S. and Woodward, D. E. (2002)
Managing smile risk, WILMOTT Magazine, September, pp. 84–108.
\item Hamilton, J. D. (1994). Time series analysis (Vol. 2). Princeton university
press.
\item Harrison, J. M. (1985). Brownian motion and stochastic flow systems
(pp. 89-91). New York: Wiley.
\item Hastings, W. K. (1970). Monte Carlo sampling methods using Markov
chains and their applications. Biometrika, 57(1), 97-109.
\item Haug, E. G. (2001). Closed form valuation of American barrier options.
International Journal of Theoretical and Applied Finance, 4(02), 355-359.
\item Haugh, M. B., \& Kogan, L. (2004). Pricing American options: a duality
approach. Operations Research, 52(2), 258-270.
\item Heath, D., Jarrow, R., \& Morton, A. (1992). Bond pricing and the
term structure of interest rates: A new methodology for contingent
claims valuation. Econometrica: Journal of the Econometric Society,
77-105.
\item Henderson, V., \& Wojakowski, R. (2002). On the equivalence of floating
and fixed-strike Asian options. Journal of Applied Probability, 39(2),
391-394.
\item Heston, S. L. (1993). A closed-form solution for options with stochastic
volatility with applications to bond and currency options. The review
of financial studies, 6(2), 327-343.
\item Ho, T. S., \& Lee, S. B. (1986). Term Structure Movements and Pricing
Interest Rate Contingent Claims. Journal of Finance, 41(5), 1011-29.
\item Hucki, Z., \& Kolokoltsov, V. N. (2007). Pricing of rainbow options:
game theoretic approach. International Game Theory Review, 9(02),
215-242.
\item Hui, C. H. (1996). One-touch double barrier binary option values.
Applied Financial Economics, 6(4), 343-346.
\item Hull, J., \& White, A. (1990). Pricing interest-rate-derivative securities.
Review of financial studies, 3(4), 573-592.
\item Hull, J., \& White, A. (1993). One-factor interest-rate models and
the valuation of interest-rate derivative securities. Journal of financial
and quantitative analysis, 28(2).
\item Hull, J. C. (2010). Options, Futures, and Other Derivatives, 7/e (With
CD). Pearson Education India.
\item Ingersoll, J. E. (1998). Approximating American options and other
financial contracts using barrier derivatives. Journal of Computational
Finance, 2(1), 85-112.
\item Ingersoll, Jr, J. E. (2000). Digital Contracts: Simple Tools for Pricing
Complex Derivatives. The Journal of Business, 73(1), 67-88.
\end{doublespace}
\item Jeong, D., Yoo, M., Yoo, C., \& Kim, J. (2019). A hybrid monte carlo
and finite difference method for option pricing. Computational Economics,
53(1), 111-124.
\begin{doublespace}
\item Jones, C. M., \& Lamont, O. A. (2002). Short-sale constraints and
stock returns. Journal of Financial Economics, 66(2), 207-239.
\item Karatzas, I., \& Wang, H. (2000). A barrier option of American type.
Applied Mathematics and Optimization, 42(3), 259-279.
\item Kashyap, R. (2016a). Securities Lending Strategies, Exclusive Valuations
and Auction Bids. Social Science Research Network (SSRN), Working
Paper.
\item Kashyap, R. (2016b). Securities Lending Strategies, TBR and TBR (Theoretical
Borrow Rate and Thoughts Beyond Rates). Social Science Research Network
(SSRN), Working Paper.
\item Kashyap, R. (2017). Notes on Uncertainty, Unintended Consequences
and Everything Else. Social Science Research Network (SSRN), Working
Paper.
\item Khouja, M. (1999). The single-period (news-vendor) problem: literature
review and suggestions for future research. Omega, 27(5), 537-553.
\item Kliger, D., \& Levikson, B. (1998). Pricing insurance contracts—an
economic viewpoint. Insurance: Mathematics and Economics, 22(3), 243-249.
\item Kou, S. G. (2002). A jump-diffusion model for option pricing. Management
science, 48(8), 1086-1101.
\item Lai, T. L., \& Xing, H. (2008). Statistical models and methods for
financial markets. New York: Springer.
\item Lambrecht, B., \& Perraudin, W. (2003). Real options and preemption
under incomplete information. Journal of Economic dynamics and Control,
27(4), 619-643.
\item Lee, S. H., \& Makhija, M. (2009). The effect of domestic uncertainty
on the real options value of international investments. Journal of
International Business Studies, 40(3), 405-420.
\item Levy, D. L. (1995). International sourcing and supply chain stability.
Journal of international business studies, 26(2), 343-360.
\item Litterman, R. B., \& Scheinkman, J. (1991). Common factors affecting
bond returns. The Journal of Fixed Income, 1(1), 54-61.
\end{doublespace}
\item Liu, X., \& Ronn, E. I. (2020). Using the binomial model for the valuation
of real options in computing optimal subsidies for Chinese renewable
energy investments. Energy Economics, 87, 104692.
\begin{doublespace}
\item Lo, C. F. (2012). The sum and difference of two lognormal random variables.
Journal of Applied Mathematics, 2012.
\end{doublespace}
\item Locatelli, G., Mancini, M., \& Lotti, G. (2020). A simple-to-implement
real options method for the energy sector. Energy, 197, 117226.
\begin{doublespace}
\item Longstaff, F. A., \& Schwartz, E. S. (2001). Valuing American options
by simulation: a simple least-squares approach. Review of Financial
studies, 14(1), 113-147.
\end{doublespace}
\item Loulianou, S. P., Leiblein, M. J., \& Trigeorgis, L. (2021). Multinationality,
portfolio diversification, and asymmetric MNE performance: The moderating
role of real options awareness. Journal of International Business
Studies, 52(3), 388-408.
\begin{doublespace}
\item Margrabe, W. (1978). The value of an option to exchange one asset
for another. The journal of finance, 33(1), 177-186.
\item Martello, S., \& Toth, P. (1987). Algorithms for knapsack problems.
Surveys in combinatorial optimization, 31, 213-258.
\item Martínez-Ceseña, E. A., \& Mutale, J. (2011). Application of an advanced
real options approach for renewable energy generation projects planning.
Renewable and Sustainable Energy Reviews, 15(4), 2087-2094.
\item Mascarenhas, B. (1982). Coping with uncertainty in international business.
Journal of International Business Studies, 13(2), 87-98.
\item Matloff, N. (2011). The art of R programming: A tour of statistical
software design. No Starch Press.
\item Merton, R. C. (1973). Theory of rational option pricing. The Bell
Journal of economics and management science, 141-183.
\item Merton, R. C. (1976). Option pricing when underlying stock returns
are discontinuous. Journal of financial economics, 3(1-2), 125-144.
\item Michaels, A., Close, A., Malmquist, D., \& Knap, A. (1997). Climate
science and insurance risk. Nature, 389(6648), 225-227.
\item Michel-Kerjan, E., \& Kunreuther, H. (2011). Redesigning flood insurance.
Science, 333(6041), 408-409.
\item Milevsky, M. A., \& Posner, S. E. (1998). Asian options, the sum of
lognormals, and the reciprocal gamma distribution. Journal of financial
and quantitative analysis, 33(03), 409-422.
\item Miller, K. D. (1992). A framework for integrated risk management in
international business. Journal of international business studies,
23(2), 311-331.
\item Mills, E. (2005). Insurance in a climate of change. Science, 309(5737),
1040-1044.
\item Mills, E. (2012). The greening of insurance. Science, 338(6113), 1424-1425.
\item Miranda, M. J., \& Fackler, P. L. (2002). Applied Computational Economics
and Finance.
\item Moraleda, J. M., \& Vorst, T. C. (1997). Pricing American interest
rate claims with humped volatility models. Journal of Banking \& Finance,
21(8), 1131-1157.
\end{doublespace}
\item Nadarajah, S., Margot, F., \& Secomandi, N. (2017). Comparison of
least squares Monte Carlo methods with applications to energy real
options. European Journal of Operational Research, 256(1), 196-204.
\item Narassimhan, E., Gallagher, K. S., Koester, S., \& Alejo, J. R. (2018).
Carbon pricing in practice: A review of existing emissions trading
systems. Climate Policy, 18(8), 967-991.
\begin{doublespace}
\item Nie, H., \& Chen, S. (2007). Lognormal sum approximation with type
IV Pearson distribution. IEEE Communications Letters, 11(10), 790-792.
\item Norstad, J. \textquotedbl The normal and lognormal distributions.\textquotedbl{}
(1999).
\end{doublespace}
\item Nunes, J. P. V., Ruas, J. P., \& Dias, J. C. (2020). Early exercise
boundaries for American-style knock-out options. European Journal
of Operational Research, 285(2), 753-766.
\begin{doublespace}
\item Oksendal, B. (2013). Stochastic differential equations: an introduction
with applications. Springer Science \& Business Media.
\item Ouwehand, P., \& West, G. (2006). Pricing rainbow options. Wilmott
magazine, 5, 74-80.
\item Paradis, E. (2002). R for Beginners.
\item Phillips, P. C., \& Yu, J. (2009). Maximum likelihood and Gaussian
estimation of continuous time models in finance. In Handbook of financial
time series (pp. 497-530). Springer Berlin Heidelberg.
\item Pirjol, D., \& Zhu, L. (2016). Discrete sums of geometric Brownian
motions, annuities and asian options. Insurance: Mathematics and Economics.
\item Qin, Y., Wang, R., Vakharia, A. J., Chen, Y., \& Seref, M. M. (2011).
The newsvendor problem: Review and directions for future research.
European Journal of Operational Research, 213(2), 361-374.
\item Quigg, L. (1993). Empirical testing of real option‐pricing models.
The Journal of Finance, 48(2), 621-640.
\item Reeb, D. M., Kwok, C. C., \& Baek, H. Y. (1998). Systematic risk of
the multinational corporation. Journal of International Business Studies,
29(2), 263-279.
\item Reuter, W. H., Szolgayová, J., Fuss, S., \& Obersteiner, M. (2012).
Renewable energy investment: Policy and market impacts. Applied Energy,
97, 249-254.
\item Rich, D. R. (1994). The mathematical foundations of barrier option-pricing
theory. Advances in Futures and Options Research, 7, 267-311.
\item Ritchken, P., \& Sankarasubramanian, L. (1995). Volatility structures
of forward rates and the dynamics of the term structure. Mathematical
Finance, 5(1), 55-72.
\item Rogers, L. C. G., \& Shi, Z. (1995). The Value of an Asian Option.
Journal of Applied Probability, 32(4), 1077-1088.
\item Rogers, L. C. (2002). Monte Carlo valuation of American options. Mathematical
Finance, 12(3), 271-286.
\item Rubin, J. D. (1996). A model of intertemporal emission trading, banking,
and borrowing. Journal of Environmental Economics and Management,
31(3), 269-286.
\item Rubinstein, M., \& Reiner, E. (1991). Breaking down the barriers.
Risk, 4(8), 28-35.
\item Rugman, A. M. (1976). Risk reduction by international diversification.
Journal of International Business Studies, 7(2), 75-80.
\item Shapiro, A., \& Homem-de-Mello, T. (2000). On the rate of convergence
of optimal solutions of Monte Carlo approximations of stochastic programs.
SIAM journal on optimization, 11(1), 70-86.
\item Shieh, L. S., Wang, H., \& Yates, R. E. (1980). Discrete-continuous
model conversion. Applied Mathematical Modelling, 4(6), 449-455.
\item Sherman, R. P., Ho, Y. Y. K., \& Dalal, S. R. (1999). Conditions for
convergence of Monte Carlo EM sequences with an application to product
diffusion modeling. The Econometrics Journal, 2(2), 248-267.
\item Shreve, S. E. (2004). Stochastic calculus for finance II: Continuous-time
models (Vol. 11). Springer Science \& Business Media.
\item Stentoft, L. (2004). Convergence of the least squares Monte Carlo
approach to American option valuation. Management Science, 50(9),
1193-1203.
\item Strike, V. M., Gao, J., \& Bansal, P. (2006). Being good while being
bad: Social responsibility and the international diversification of
US firms. Journal of International Business Studies, 37(6), 850-862.
\item Stulz, R. (1982). Options on the minimum or the maximum of two risky
assets: analysis and applications. Journal of Financial Economics,
10(2), 161-185.
\item Sung, S. W., Lee, J., \& Lee, I. B. (2009). Process identification
and PID control. John Wiley \& Sons.
\item Sweeney, J., \& Sweeney, R. J. (1977). Monetary theory and the great
Capitol Hill Baby Sitting Co-op crisis: comment. Journal of Money,
Credit and Banking, 9(1), 86-89.
\item Taleb, N. (2005). Fooled by randomness: The hidden role of chance
in life and in the markets (Vol. 1). Random House Incorporated.
\item Taleb, N. N. (2007). The black swan: The impact of the highly improbable.
Random house.
\item Tee, J., Scarpa, R., Marsh, D., \& Guthrie, G. (2014). Forest valuation
under the New Zealand emissions trading scheme: a real options binomial
tree with stochastic carbon and timber prices. Land Economics, 90(1),
44-60.
\item Thavaneswaran, A., Appadoo, S. S., \& Frank, J. (2013). Binary option
pricing using fuzzy numbers. Applied Mathematics Letters, 26(1), 65-72.
\item Tong, T. W., \& Reuer, J. J. (2007). Real options in multinational
corporations: Organizational challenges and risk implications. Journal
of International Business Studies, 38(2), 215-230.
\item Trigeorgis, L. (1993). The nature of option interactions and the valuation
of investments with multiple real options. Journal of Financial and
quantitative Analysis, 28(01), 1-20.
\item Trigeorgis, L. (2005). Making use of real options simple: An overview
and applications in flexible/modular decision making. The Engineering
Economist, 50(1), 25-53.
\end{doublespace}
\item Trigeorgis, L., \& Tsekrekos, A. E. (2018). Real options in operations
research: A review. European Journal of Operational Research, 270(1),
1-24.
\begin{doublespace}
\item Uboldi, A. (2016). Are derivatives introducing distortions in agricultural
markets?. Agricultural Markets Instability: Revisiting the Recent
Food Crises, 117-123.
\item Vasicek, O. (1977). An equilibrium characterization of the term structure.
Journal of financial economics, 5(2), 177-188.
\item Venables, W. N., Smith, D. M., \& R Development Core Team. (2016).
An introduction to R.
\end{doublespace}
\item Wang, C., \& Chen, X. (2017). Option pricing and coordination in the
fresh produce supply chain with portfolio contracts. Annals of Operations
Research, 248(1-2), 471-491.
\item Wen, F., Wu, N., \& Gong, X. (2020). China's carbon emissions trading
and stock returns. Energy Economics, 86, 104627.
\begin{doublespace}
\item West, G. (2005). Calibration of the SABR model in illiquid markets.
Applied Mathematical Finance, 12(4), 371-385.
\item Windcliff, H., Wang, J., Forsyth, P. A., \& Vetzal, K. R. (2007).
Hedging with a correlated asset: Solution of a nonlinear pricing PDE.
Journal of Computational and Applied Mathematics, 200(1), 86-115.
\item Xu, M. (2006). Risk measure pricing and hedging in incomplete markets.
Annals of Finance, 2(1), 51-71.
\item Yan, S. (2011). Jump risk, stock returns, and slope of implied volatility
smile. Journal of Financial Economics, 99(1), 216-233.
\end{doublespace}
\item Zanger, D. Z. (2018). Convergence of a least‐squares Monte Carlo algorithm
for American option pricing with dependent sample data. Mathematical
Finance, 28(1), 447-479.
\begin{doublespace}
\item Zhu, S. P. (2006). An exact and explicit solution for the valuation
of American put options. Quantitative Finance, 6(3), 229-242.
\item Zimmermann, H. J. (1996). Fuzzy Control. In Fuzzy Set Theory—and Its
Applications (pp. 203-240). Springer Netherlands.
\item Zvan, R., Vetzal, K. R., \& Forsyth, P. A. (2000). PDE methods for
pricing barrier options. Journal of Economic Dynamics and Control,
24(11), 1563-1590.
\end{doublespace}
\end{enumerate}
\begin{doublespace}

\section{\label{sec:Dictionary-of-Notation}Appendix A: Dictionary of Notation
and Terminology for Term Loans}
\end{doublespace}
\begin{itemize}
\begin{doublespace}
\item $H_{i}$, the quantity of the term deal for security $i$ for the
entire duration of the loan.
\item $A_{it}$, the total Availability on the stock in shares, at a particular
time, $t$, for security, $i$. This is the amount available to make
term loans and includes fresh supply reported from external lenders,
the existing borrow book positions, inventory from internal trading
desks and amounts taken and also unused holdings from exclusive arrangements. 
\item $B_{it}$, the Borrow book carried by the desk in shares, at a particular
time, $t$, for security, $i$. This is the existing amount borrowed
form external lenders.
\item $I_{it}$, the Internal Inventory the intermediary holds in shares,
at a particular time, $t$, for security, $i$.
\item $O_{it}$, the additional supply that can be sourced from beneficial
owners Other than exclusives in shares, at a particular time, $t$,
for security, $i$.
\item $E_{it}$, the amount taken out from the Exclusive pool in shares,
at a particular time, $t$, for security, $i$.
\item $P_{it}$, the holdings available in the exclusive Pool that is not
current drawn or is unused in shares, at a particular time, $t$,
for security, $i$.
\item $L_{it}$, the Loan book carried by the desk in shares, at a particular
time, $t$, for security, $i$. This is the existing amount of loans
to external borrowers.
\item $R_{it}$, the Rate on the loan charged by the intermediary, at a
particular time, $t,$ until the next time period, $t+1$, for security,
$i.$
\item $Q_{it}$, the borrow rate, an alternate rate to $R_{it}$, indicating
the cost at a particular time, t, until the next time period, $t+1$,
for security, $i$. This could be the rate at which supply from other
beneficial owners is sourced or could be a theoretical rate when no
rate from other beneficial owners is available. $Q_{it}<R_{it}$.
\item $q_{it}$, is the spread added to the borrow rate to form the loan
rate $R_{it}$.
\item $c,$ the constant factor governing the spread and the borrow rate.
\item $r$, the interest rate or other applicable rate.
\item $f$, the constant fee to utilize the exclusive holdings. This is
converted from the payment made for the duration of the exclusive
to apply on a daily basis.
\item $S_{it}$, the Security price at a particular time, $t,$ until the
next time period, $t+1$, for security, $i.$
\item $\upsilon,$ the Valuation of the term loan, for the duration extending
from $t=0\;to\;t=T.$
\item $w,$ the Valuation of the term loan expressed as a continuous compounded
interest rate, for the duration extending from $t=0\;to\;t=T.$
\item $T$, is the expiration time period or the total duration for which
the term loan will be contracted.
\item \textcolor{black}{$\left\{ \upsilon^{constant},\upsilon^{proportional\;time},\upsilon^{constant\;counter},\upsilon^{proportional\;counter},\right.$}
\end{doublespace}

\begin{doublespace}
\textcolor{black}{$\left.\upsilon^{stock\;holding},\upsilon^{borrow\;rate},\upsilon^{stochastic\;demand},\upsilon^{historical}\right\} $,
is the set of valuations.}
\end{doublespace}
\begin{doublespace}
\item $\beta=\frac{1}{\left(1+r\right)},$ is the discount factor, $r$
is the risk free rate of interest. Further complications can be introduced
by incorporating continuous time extensions to the short rate process. 
\item $n,$ the number of securities available in the term deal, $i\in\left\{ 1,\,...\,,n\right\} $.
This applies when we are looking at rainbows, baskets or other multi-security
structures.
\item $K_{i},\tilde{K}_{i},$ the constant payoffs on the term deal if the
barrier is breached for security $i$ in the downward and in the upward
direction.
\item $\theta_{ij},\lambda_{ij}$ are the scaling constants for security
$i$ in the partition $j$ for the availability and stock price time
series. These facilitate the use of the historical time series for
option valuation.
\item $p$ is the number of partitions of the historical time series for
option valuation.
\item $T_{js}$ and $T_{je}$ are the start and end times of the $j^{th}$
historical time series.
\item $D_{it}$, the Locate requests received, in shares, at a particular
time, $t$, for security, $i$.
\item $\delta_{it}\in\left[0,1\right],$ the conversion rate of locates
into borrows, at a particular time, $t$, for security, $i$. We can
simplify this to be the same per security.
\item $\delta_{i},$ the conversion rate of locates into borrows for security,
$i$. We can simplify this further to be a constant across time and
securities, $\delta$.
\item $\delta_{i}$$D_{it}$, then indicates the excess demand that the
desk receives, in shares, at a particular time, $t$, for security,
$i$.
\item $P$, the profits of the loan desk from the entire loan book over
the duration $T$.
\item $N,$ the number of trading intervals.
\item The length of each trading interval, $\tau=T/N$. We assume the time
intervals are of the same duration, but this can be relaxed quite
easily. In continuous time, this becomes, $N\rightarrow\infty,\tau\rightarrow0$.
\item The time then becomes divided into discrete intervals, $t_{k}=k\tau,\;k=0,...,N$.
We simplify this and write it as $t=0\;to\;t=T$ with unit increments.
\item It is common practice to consider daily increments in time for one
year. The fees paid generally also applies on weekends and holidays,
though there would be no change in any of the variables on these days.
Some firms use 252 trading days to annualize daily loan rates and
other fee terms. 
\item $\forall$ $a$, $b$ $\in\Re_{+}$, $\quad$$a^{+}=\max\left\{ a,0\right\} $,
$\quad$$a^{-}=\max\left\{ -a,0\right\} $, and$\quad$ $a\bigwedge b=\min\left\{ a,b\right\} $.
\end{doublespace}
\end{itemize}
\begin{doublespace}

\section{Appendix B: \label{subsec:Proof-of-Proposition:The-indicator-valuation-expression}Proof
of Proposition \ref{The-indicator-valuation-expression}}
\end{doublespace}
\begin{proof}
\begin{doublespace}
It is worth noting that $\text{Up}_{it}$ and $\text{Down}_{it}$
are mutually exclusive. Only one of them can be one in a given time
period. We consider the following four scenarios that can happen,
back to back, or in successive time periods. 
\[
\left[\left\{ \text{Up}_{i,t-1},\text{Up}_{it}\right\} \left\{ \text{Up}_{i,t-1},\text{Down}_{it}\right\} \left\{ \text{Down}_{i,t-1},\text{Down}_{it}\right\} \left\{ \text{Down}_{i,t-1},\text{Up}_{it}\right\} \right]
\]
Of the above scenarios, the following indicates the payoff incurred
correspondingly. There is a cashflow exchanged, when a state change
occurs either from Up to Down or from Down to Up.
\[
\left[\left\{ 0\right\} \left\{ K_{i}\right\} \left\{ 0\right\} \left\{ \tilde{K}_{i}\right\} \right]
\]
The above is equivalent to 
\[
\left[\left\{ \text{Up}_{it},\text{Up}_{i,t-1}\right\} \left\{ \text{Up}_{it},\text{Down}_{i,t-1}\right\} \left\{ \text{Down}_{it},\text{Down}_{i,t-1}\right\} \left\{ \text{Down}_{it},\text{Up}_{i,t-1}\right\} \right]\equiv\left[\left\{ 0\right\} \left\{ \tilde{K}_{i}\right\} \left\{ 0\right\} \left\{ K_{i}\right\} \right]
\]
The table below (Figure \ref{fig:Transacion-Cost-Table}) summarizes
the payoffs exchanged, based on the difference between variables across
successive time periods, when one of the four combinations occurs.
As an example, $\left\{ \text{Down}_{it},\text{Up}_{i,t-1}\right\} $
means that in time period $t-1$, security $i$ is in the $\text{Up State}$,
or $\text{Up}_{i,t-1}=1$ and in time period $t$ it is in the $\text{Down State}$,
or $\text{Down}_{it}=1$. Hence, when this combination occurs, we
have, $\left(\text{Down}_{it}-\text{Down}_{i,t-1}\right)=1$ and $\left(\text{Up}_{it}-\text{Up}_{i,t-1}\right)=-1$.

{\footnotesize{}}
\begin{figure}[H]
\begin{tabular}{|l|c|c|c|c|}
\hline 
 & $\left\{ \text{Up}_{it},\text{Up}_{i,t-1}\right\} $ & $\left\{ \text{Up}_{it},\text{Down}_{i,t-1}\right\} $ & $\left\{ \text{Down}_{it},\text{Down}_{i,t-1}\right\} $ & $\left\{ \text{Down}_{it},\text{Up}_{i,t-1}\right\} $\tabularnewline
\hline 
\hline 
$\left(\text{Up}_{it}-\text{Up}_{i,t-1}\right)$ & $0$ & $1$ & $0$ & $-1$\tabularnewline
\hline 
$\left(\text{Up}_{it}-\text{Down}_{i,t-1}\right)$ & $1$ & $0$ & $-1$ & $0$\tabularnewline
\hline 
$\left(\text{Down}_{it}-\text{Down}_{i,t-1}\right)$ & $0$ & $-1$ & $0$ & $1$\tabularnewline
\hline 
$\left(\text{Down}_{it}-\text{Up}_{i,t-1}\right)$ & $-1$ & $0$ & $1$ & $0$\tabularnewline
\hline 
\end{tabular}

\caption{{\footnotesize{}\label{fig:Transacion-Cost-Table}}Transaction Cost
Table}
\end{figure}
A downward transition occurs, whenever in the sequence of observations
for the variables in the $\text{Up State}$ we get a $\ldots10\ldots$,
that is, $\ldots\text{Up}_{it}\text{Up}_{i,t+1}\ldots\equiv\ldots10\ldots\mid t\in\left\{ 0,1,2,\ldots,T-1\right\} $.
Similarly, an upward transition occurs, whenever in the sequence of
observations for the variables in the $\text{Down State}$ we get
a $\ldots10\ldots$, that is, $\ldots\text{Down}_{it}\text{Down}_{i,t+1}\ldots\equiv\ldots10\ldots\mid t\in\left\{ 0,1,2,\ldots,T-1\right\} $.
An upward transition is also equivalent to getting a $\ldots01\ldots$
for the variables in the $\text{Up State}$, that is, $\ldots\text{Up}_{it}\text{Up}_{i,t+1}\ldots\equiv\ldots01\ldots\mid t\in\left\{ 0,1,2,\ldots,T-1\right\} $.
From this, we can count the number of times a $\text{Up State}$ to
$\text{Down State}$ transition happens as,
\begin{eqnarray*}
\text{Number of Downward Transitions} & = & \left|\left\{ t\in\left\{ 1,2,\ldots,T-1\right\} \mid\text{Up}_{it}=1\text{ and }\text{Down}_{i,t+1}=1\right\} \right|\\
 & = & E_{0}\left[\sum_{i=1}^{n}\sum_{t=0}^{T-1}\left\lfloor \frac{1+\text{Up}_{it}-\text{Up}_{i,t+1}}{2}\right\rfloor \right]
\end{eqnarray*}
\begin{eqnarray*}
\text{Number of Upward Transitions} & = & \left|\left\{ t\in\left\{ 1,2,\ldots,T-1\right\} \mid\text{Down}_{it}=1\text{ and }\text{Up}_{i,t+1}=1\right\} \right|\\
 & = & E_{0}\left[\sum_{i=1}^{n}\sum_{t=0}^{T-1}\left\lfloor \frac{1-\text{Up}_{it}+\text{Up}_{i,t+1}}{2}\right\rfloor \right]
\end{eqnarray*}
Here, $\left|\left\{ A\right\} \right|$ indicates the cardinality
of the set $A$. The valuation expressions are given below, keeping
in mind that if we start in the $\text{Down State}$ state a cashflow
applies.

\begin{eqnarray*}
\upsilon^{constant\;counter} & = & E_{0}\left[K_{i}\sum_{i=1}^{n}\left(1-\text{Up}_{i0}\right)+K_{i}\sum_{i=1}^{n}\sum_{t=0}^{T-1}\left\lfloor \frac{1+\text{Up}_{it}-\text{Up}_{i,t+1}}{2}\right\rfloor \right.\\
 &  & \left.-\tilde{K}_{i}\sum_{i=1}^{n}\sum_{t=0}^{T-1}\left\lfloor \frac{1-\text{Up}_{it}+\text{Up}_{i,t+1}}{2}\right\rfloor \right]
\end{eqnarray*}
\begin{eqnarray*}
\upsilon^{proportional\;counter} & = & E_{0}\left[K_{i}T\sum_{i=1}^{n}\left(1-\text{Up}_{i0}\right)+K_{i}\sum_{i=1}^{n}\sum_{t=0}^{T-1}\left(T-t-1\right)\left\lfloor \frac{1+\text{Up}_{it}-\text{Up}_{i,t+1}}{2}\right\rfloor \right.\\
 &  & \left.-\tilde{K}_{i}\sum_{i=1}^{n}\sum_{t=0}^{T-1}\left(T-t-1\right)\left\lfloor \frac{1-\text{Up}_{it}+\text{Up}_{i,t+1}}{2}\right\rfloor \right]
\end{eqnarray*}
\end{doublespace}
\end{proof}

\end{document}